\documentclass[journal]{IEEEtran}
\usepackage[font={small}]{caption}
\usepackage{color}
\usepackage{cite}
\usepackage{graphicx}
\usepackage{epstopdf}
\usepackage{subfig}
\usepackage{amsthm}
\usepackage{tabularx}
\usepackage{amssymb}
\newtheorem{theorem}{Theorem}
\newtheorem{lemma}[theorem]{Lemma}
\usepackage{array,booktabs,arydshln,xcolor}

\usepackage{epsfig}
\usepackage{amsfonts}
\usepackage{epstopdf}
\usepackage{mathdots}
\usepackage{balance}
\usepackage{pseudocode}
\usepackage[noend]{algpseudocode}
\usepackage{algorithm}
\usepackage{amsmath}
\usepackage{listings}
\usepackage[keeplastbox]{flushend}
\usepackage{fancybox}
\usepackage{epsfig}
\usepackage{mathtools}

\ifCLASSINFOpdf
\else
\fi
\newcommand{\comm}[1]{}

\allowdisplaybreaks
\begin{document}

\title{Swift-Link: A compressive beam alignment algorithm for practical mmWave radios}
\author{{\IEEEauthorblockN{Nitin Jonathan Myers, {\it Student Member, IEEE}, Amine Mezghani, {\it Member, IEEE},\\  and Robert W. Heath Jr., {\it Fellow, IEEE}. }}
\thanks{ N. J. Myers (nitinjmyers@utexas.edu), A. Mezghani (amine.mezghani@utexas.edu) and R. W. Heath Jr. (rheath@utexas.edu) are with the  Wireless Networking and Communications  Group, The University of Texas at Austin, Austin,
TX 78712 USA. This research was partially supported by the National Science
Foundation under grant numbers NSF-CNS-1702800, NSF-CNS-1731658, and by the U.S. Department of Transportation through the Data-Supported Transportation Operations and Planning (D-STOP) Tier 1 University Transportation Center.}}

\maketitle
\begin{abstract}
Next generation wireless networks will exploit the large amount of spectrum available at millimeter wave (mmWave) frequencies. Design of mmWave systems, however, is challenging due to strict power, cost and hardware constraints at higher bandwidths. To achieve a good SNR for communication, mmWave systems use large antenna arrays. Beamforming with highly directional beams is one way to use the antennas. As the channel changes over time, the beams that maximize the SNR have to be estimated quickly to reduce the training overhead. Prior work has exploited the observation that mmWave channels are sparse to perform compressed sensing (CS) based beam alignment with few channel measurements. Most of the existing CS-based algorithms, however, assume perfect synchronization and fail in the presence of carrier frequency offset (CFO). This paper presents Swift-Link, a fast beam alignment algorithm that is robust against the offset. Swift-Link includes a novel randomized beam training sequence that minimizes the beam alignment errors due to CFO and a low-complexity algorithm that corrects these errors. Even with strict hardware constraints, our algorithm uses fewer channel measurements than comparable CS algorithms and has analytical guarantees. Swift-Link requires a small output dynamic range at the analog-to-digital converter compared to beam-scanning techniques. 
\end{abstract}

\begin{IEEEkeywords} 
Compressive beam training, carrier frequency offset, robust beam alignment, low-complexity sparse self-calibration, robust compressive sensing, mm-wave.
\end{IEEEkeywords}
\IEEEpeerreviewmaketitle

\section{Introduction}
Multiple-input multiple-output (MIMO) communication at millimeter wave is very different from that in the  lower frequency systems \cite{ranganmmwave}. To achieve sufficent link margin, large antenna arrays have to be used at the transmitter (TX) and the receiver (RX) \cite{mmintro}. As a result, mmWave MIMO channels  have a higher dimension compared to typical lower frequency systems. Furthermore, cost and power consumption are major issues at the large bandwidths that accompany mmWave \cite{heathoverview}, and thereby restrict the use of conventional MIMO architectures. The phased array architecture \cite{heathoverview} is one hardware design that limits power consumption by using fewer radio-frequency chains compared to the number of antennas. As a result, the phased array can only obtain a lower dimensional projection of the MIMO channel for a given phase shift configuration. The compressive nature of the hardware and the use of large antenna arrays complicate tasks like channel estimation or beam alignment in mmWave systems.
\par Most of the beam alignment techniques that exploit sparsity of mmWave channels fail if applied to a system where CFO has not been corrected \cite{agile}. As mmWave MIMO channels are approximately sparse in a certain dictionary \cite{heathoverview}, several algorithms that exploit sparsity have been proposed to perform channel estimation or beam alignment with few measurements \cite{cschest, kiranchannel}. Most of them are based on standard compressed sensing (CS) \cite{csintro} and are vulnerable to CFO. CFO arises due to the mismatch in the carrier frequencies of the local oscillators at the TX and the RX \cite{heathwicomm}. The offset corrupts the phase of the channel measurements and thereby CS algorithms that do not model CFO perform poorly. Therefore, joint CFO and channel estimation algorithms or beam alignment algorithms that are robust to CFO must be developed.
\par Beam alignment techniques like the exhaustive beam search and hierarchical search \cite{exhhier}  are robust to CFO as they use just the magnitude of the channel measurements. On the one hand, exhaustive beam search scans all possible directions and results in large training overhead. On the other hand, hierarchical search suffers from low SNR prior to beamforming and is susceptible to multipath interference within a sector. Furthermore, it requires feedback on every decision to continue with the search process in transmit beam alignment \cite{agile}. The limitations of exhaustive and hierarchical search motivate the need to develop beam alignment algorithms that can exploit the sparsity of mmWave channels and are robust to CFO. 
\par Compressive beam alignment algorithms that are robust to CFO were proposed in \cite{agile},\cite{nitinanalog} and \cite{Javi} for the phased array architecture. In our prior work \cite{nitinanalog}, we modeled both CFO and the mmWave channel in a tensor, and exploited sparsity of the tensor to perform joint CFO and channel estimation. The sparsity-aware algorithm in \cite{nitinanalog} can also be applied to joint CFO estimation and beam alignment. The approach in \cite{nitinanalog}, however, increases the number of optimization variables and results in high computational complexity. A multi-stage technique that overcomes the high complexity of the tensor based approach was proposed in \cite{Javi} for joint CFO and channel estimation. A non-coherent compressive beam alignment strategy was developed in \cite{agile} using hashing techniques. The energy based algorithm in \cite{agile}, however, assumes high SNR and ignores the phase of the received samples.  
\par In this paper, we optimize the CS measurement matrix for beam alignment robust to CFO, and present a low complexity algorithm called Swift-Link. We assume that there is perfect frame timing synchronization between the TX and the RX. We summarize the main contributions of our work as follows.
\begin{itemize}
\item For planar phased arrays, we define a class of training matrices that use Zadoff-Chu (ZC)  sequences along the spatial dimension. We show that CS using the ZC-based training is equivalent to partial DFT CS of a transformed channel. With this equivalence, we derive reconstruction guarantees for compressive channel estimation under phased array constraints. Furthermore, we  propose the concept of a trajectory to develop a tractable approach for CS-based beam alignment robust to CFO. 
\item We design a novel randomized beam training sequence within the class of ZC-based training such that CFO induces bounded shifts in the beamspace channel obtained using standard CS algorithms. Such shifts can be ignored when the CFO is small relative to the angular resolution of the array. For other cases, we provide a low-complexity algorithm to correct the shifts. 
\item We determine a sufficient condition on the number of channel measurements for successful beam training using our algorithm. For an $N \times N$ uniform planar array, our analysis suggests $\mathcal{O}(N \mathrm{log} N)$ channel measurements suffice for CFO robust  compressive channel estimation using Swift-Link. Our analysis is also valid for the off-grid case and can be extended to generic CS problems with structured errors in the CS measurements.   
\item Using simulations, we evaluate the performance of our algorithm for fast beam alignment in  a practical wideband mmWave setting. Swift-Link estimates the MIMO channel and CFO compressively, and performs better than comparable algorithms that are robust to CFO.
\end{itemize}
Swift-Link performs compressive channel estimation robust to CFO using the optimized CS measurement matrix. The estimated channel is then used for beam alignment. While we assume an analog beamforming setup, Swift-Link can be extended to switching and hybrid beamforming architectures in a straightforward manner. 
\par \textbf{Notation}$:$ $\mathbf{A}$ is a matrix, $\mathbf{a}$ is a column vector and $a, A$ denote scalars. Using this notation $\mathbf{A}^T,\overline{\mathbf{A}}$ and $\mathbf{A}^{\ast} $ represent the transpose, conjugate and conjugate transpose of $\mathbf{A}$. We use $\mathrm{diag}\left(\mathbf{a}\right)$ to denote a diagonal matrix with entries  of $\mathbf{a}$ on its diagonal. The scalar $a\left[m \right]$ denotes the $m^{\mathrm{th}}$ element of $\mathbf{a}$, and $\mathbf{A}\left(k,\ell\right)$ is the entry of $\mathbf{A}$ in the $k^{\mathrm{th}}$ row and ${\ell}^{\mathrm{th}}$ column. The inner product of two matrices $\mathbf{A}$ and $\mathbf{B}$ is defined as $\left\langle \mathbf{A},\mathbf{B}\right\rangle =\sum_{k,\ell}\mathbf{A}\left(k,\ell \right)\overline{\mathbf{B}}\left(k,\ell\right)$. The symbols $\odot$ and $\circledast$ are used for the Hadamard product and 2D circular convolution \cite{imageprocess}. The matrix $\mathbf{U}_N \in \mathbb{C}^{N \times N}$ denotes the unitary Discrete Fourier Transform (DFT) matrix. The set $\mathcal{I}_N$ denotes the set of integers $\left\{ 0,1,2,...,N-1\right\}$. The difference between two sets $\mathcal{A}$ and $\mathcal{B}$ is defined as $\mathcal{A} \setminus \mathcal{B}$. We use $\mathbf{e}_k$ to represent the $(k+1)^{\mathrm{th}}$ canonical basis vector and $\mathbf{1}$ to represent an all-ones vector. We use $\left\lceil a \right\rceil$ and $\left\lfloor a \right\rfloor $ to denote the ceil and floor of $a$. The natural logarithm of $a$ is $\mathrm{log}(a)$.
\section{System and channel model} \label{sec:syschan}
In this section, we describe the underlying hardware setup and channel model used in our work.  For simplicity of exposition, we consider a narrowband mmWave system with an analog beamforming architecture and focus on transmit beam alignment. We extend our algorithm to the wideband case in Section \ref{sec:wbextend} and provide simulation results for a practical mmWave setting in Section \ref{sec:swsim}. 
\subsection{System model} \label{sec:sysmodel}
\begin{figure}[h]
\centering
\vspace{-5mm}
\includegraphics[trim=0cm 0cm 0cm 0cm,clip=true ,width=0.5 \textwidth]{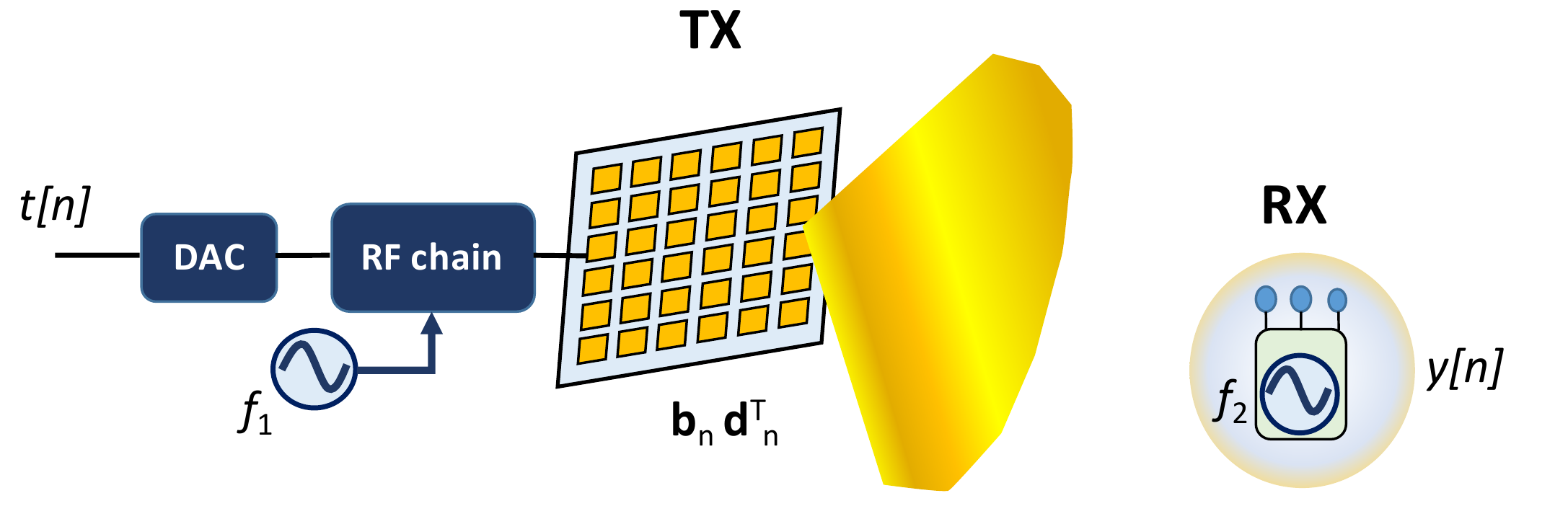}
\caption{Beam training in an analog beamforming system with local oscillators operating at $f_1$ and $f_2$. The TX transmits several beacons using an $N\times N$ phased array and the RX acquires channel measurements with a fixed quasi-omnidirectional pattern.}
  \label{fig:architect}
\end{figure}
We consider a MIMO system in which the TX is equipped with an $N \times N$ uniform planar array (UPA) of antennas as shown in Fig.~\ref{fig:architect}. All the antennas at the TX are connected to a single RF chain through phase shifters. By controlling the phase shifts, the transmit beam can be steered along the elevation and azimuth  dimensions. In this work, the set of possible phase shifts is restricted to a kronecker codebook \cite{kroncodebook}, i.e., the phase shift matrix applied to the UPA can always be decomposed as an outer product of two phase shift vectors. The decomposition is practical and allows our solution to be applied to MIMO systems with uniform linear array of antennas at the TX and the RX \cite{exhhier}. The objective of transmit beam alignment is to estimate the best set of phase shifts at the TX that maximize the received SNR. Although the RX can have multiple antennas, it is assumed to operate with a single RF chain in a fixed quasi-omnidirectional mode during beam training. Our assumption can be justified as it is the case with beam training in the IEEE 802.11ad standard \cite{11ad}.  
\par The baseband signal at the TX is up-converted to a mmWave carrier frequency $f_1$ and is transmitted through the UPA. The transmitted RF signal propagates through the wireless channel and is downconverted at the RX using a carrier frequency $f_2$, that slightly differs from $f_1$. The absolute difference, i.e., $\left| f_2-f_1 \right|$ is typically in the order of several parts per millions (ppms) of the carrier frequency $f_1$. Due to the high carrier frequencies at mmWave, even such small differences can result in significant phase change across the channel measurements in a short duration. In this work, we ignore the effect of stochastic phase noise and focus on the impact of CFO. To validate this assumption, a short-packet transmission scheme \cite{kiranchannel} is used so that the phase error induced by the random phase noise process does not change significantly across the channel measurements.    
\par Now, we describe the discrete time baseband equivalent model for the narrowband beamforming system. For a beam training slot of duration $T$ seconds, we define the digital domain CFO as $\epsilon= 2\pi \left(f_1-f_2 \right) T$. In the $n^{\mathrm{th}}$ beam training slot, let $\mathbf{d}_n \in \mathbb{C}^{N}$ and $\mathbf{b}_n \in \mathbb{C}^{N}$ be unit norm phase shift vectors corresponding to the azimuth and elevation directions. The entries of these vectors are constrained to have the same magnitude, while their phase can take values from a finite set depending on the resolution of the phase shifters. By the kronecker codebook assumption, the phase shift matrix applied to the UPA at the TX during the $n^{\mathrm{th}}$  beam training slot is $\mathbf{b}_n\mathbf{d}^T_n$. Let $t[n]$ denote the unit norm pilot symbol transmitted by the TX during the $n^{\mathrm{th}} $ slot. Let $\mathbf{H} \in \mathbb{C}^{N \times N}$ denote the discrete time baseband equivalent of the narrowband channel between the UPA at the TX and the RX. Our system model does not incorporate the antenna elements at the RX as the RX performs quasi-omnidirectional reception using a fixed beam pattern throughout beam training. Assuming perfect frame timing synchronization, the received signal in the $n^{\mathrm{th}}$ beam training slot can be expressed as 
\begin{align}
\tilde{y}[n]&=e^{ \mathrm{j}\epsilon n} \left\langle \mathbf{H}, \mathbf{b}_n \mathbf{d}^{T}_n\right\rangle t[n] + \tilde{v}[n]\\
&=\mathbf{b}^{\ast}_n \mathbf{H} \overline{\mathbf{d}}_n e^{\mathrm{j} \epsilon n} t[n] + \tilde{v}[n],
\end{align}
where $\tilde{v}[n]\sim \mathcal{CN}\left(0,\sigma^2\right)$ is additive white Gaussian noise.  As the pilot symbol $t [n]$ is known to the RX, we redefine the received samples as $y[n]=\overline{t}[n]\tilde{y}[n]$. The channel measurement for the $n^{\mathrm{th}}$ symbol duration is then 
\begin{equation}
y[n]=\mathbf{b}^{\ast}_n \mathbf{H} \overline{\mathbf{d}}_n e^{ \mathrm{j} \epsilon n} + v[n],
\label{eq:meas}
\end{equation}
where $v[n]$ has the same statistics as that of $\tilde{v}[n]$ since $t[n]$ is unit norm. The measurements in \eqref{eq:meas} are noisy projections of the channel matrix perturbed by phase errors due to CFO. Furthermore, the phase errors vary across different observations and prevent direct application of several existing channel estimation and beam training algorithms. 
\subsection{Channel model}
We consider a geometric-ray-based model for the narrowband mmWave channel when the RX performs omnidirectional reception \cite{tsewicomm}. Let $\gamma_{k}$, $\theta_{e,k}$, and $\theta_{a,k}$ denote the complex gain, elevation angle-of-departure, and azimuth angle-of-departure of the $k^{\mathrm{th}}$ ray. 
At this point, we do not make any assumption on these parameters and they can come from any distribution.
With $\omega_{e,k}=\pi\,\mathrm{sin}\, \theta_{e,k} \mathrm{sin}\, \theta_{a,k}$, $\omega_{a,k}=\pi \, \mathrm{sin}\, \theta_{e,k} \mathrm{cos}\, \theta_{a,k}$, and the Vandermonde vector
\begin{equation}
\mathbf{a}_N\left(\Delta \right)=\left[1\,, e^{\mathrm{j} \Delta}\,, e^{\mathrm{j} 2\Delta}\,, \cdots\,, e^{\mathrm{j} (N-1)\Delta}\right]^{T},
\end{equation}
the wireless channel for a half wavelength spaced UPA in the baseband is given by
\begin{equation}
\mathbf{H}=\sum_{k=1}^{K}\gamma_{k}\mathbf{a}_N \left(\omega_{e,k}\right)\mathbf{a}_N^{T}\left(\omega_{a,k}\right).
\label{eq:nbchannel}
\end{equation}
The channel is comprised of $K$ rays that are characterized by their gain and the angle of departure. 
\par Due to the propagation characteristics of the environment at mmWave frequencies, the channel is approximately sparse when expressed in an appropriate basis \cite{heathoverview}. For instance, the 2D-DFT basis is often chosen for the sparse representation of $\mathbf{H}$ \cite{beamsparse}. Let $\mathbf{X} \in \mathbb{C}^{N \times N}$ denote the 2D-DFT of $\mathbf{H}$ such that   
\begin{equation}
\mathbf{H}= \mathbf{U}^{\ast}_{N} \mathbf{X} \mathbf{U}^{\ast}_{N}.
\label{eq:mimoangledom}
\end{equation}   
The matrix $\mathbf{H}$ is the antenna domain channel and $\mathbf{X}$ is called as the beamspace channel. 
The entries of $\mathbf{X}$ are the components of $\mathbf{H}$ along the azimuth and elevation directions defined by the 2D-DFT dictionary. In practice, $\mathbf{X}$ is approximately sparse as the true angles of departure may not align with those defined by the 2D-DFT dictionary. For our analysis, we consider $\mathbf{X}$ to be perfectly sparse while our simulation results are for the realistic case where $\mathbf{X}$ is approximately sparse.
\section{Breaking the hardware barrier : Efficient CS for analog beamforming} \label{sec:RIPguaran}
\par Most CS-based algorithms for channel sparsity-aware beam training \cite{mo2017hybrid, kiranchannel,cschest} use IID random phase shifts to obtain channel measurements. Although random phase shift-based CS matrices work in practice, it is possible that other phased-array compatible CS matrices can perform better beam alignment.  
In this section, we assume perfect synchronization, i.e., $\epsilon=0$, to design efficient CS measurement matrices. 
\par  Hardware constraints in phased arrays are the bottleneck that does not allow the realization of several efficient CS matrices that satisfy the restricted isometry property (RIP). For $\epsilon=0$, the channel measurement in \eqref{eq:meas} is a projection of the channel matrix onto a beam training matrix, i.e., $\mathbf{b}_n \mathbf{d}^{T}_n$. A standard result in CS states that the ``best'' way to estimate $\mathbf{H}$ is to sense it using a basis that is maximally incoherent with the sparsity basis \cite{csintro}. It is also well known that the canonical basis, with $N \times N$ matrices of the form $\mathbf{e}_k \mathbf{e}^T_{\ell}$, is maximally incoherent \cite{csintro,gaussrip} with the 2D-DFT basis, which is used for a sparse representation of $\mathbf{H}$. The CS matrix that results from projecting a matrix with a sparse 2D-DFT representation onto fewer elements of the canonical basis is called as a partial 2D-DFT CS matrix. Projecting the channel onto the canonical basis elements, however, requires turning off $N^2-1$ antennas and cannot be implemented with phased arrays. As a result, partial 2D-DFT CS matrices that are known to satisfy the RIP cannot be realized directly in phased arrays. Similarly, IID Gaussian and IID Bernoulli CS matrices cannot be realized in phased arrays. To understand the realizability of such CS matrices, we define $\mathbf{A}_{\mathrm{CS}}$ as the CS matrix corresponding to \eqref{eq:meas} for $\epsilon=0$.  From \eqref{eq:meas} and \eqref{eq:mimoangledom}, it can be observed that the  $n^{\mathrm{th}}$ row of $\mathbf{A}_{\mathrm{CS}}$ is $\mathbf{A}_{\mathrm{CS}}(n,:)=(\mathbf{d}^{\ast}_n\mathbf{U}^{\ast}_N) \otimes (\mathbf{b}^{\ast}_n\mathbf{U}^{\ast}_N)$ \cite{cschest}. Now, realizing IID Bernoulli or IID Gaussian entries in $\mathbf{A}_{\mathrm{CS}}$  is difficult as it requires the design of unimodular vectors $\mathbf{b}_n$ and $\mathbf{d}_n$ such that the kronecker product of their Fourier transforms satisfy the desired distribution.
\par For efficient CS-based beam alignment, we use circulantly shifted versions of a fixed Zadoff-Chu sequence \cite{zadoffchu} for phase shifts along the azimuth and elevation dimensions of the UPA. ZC sequences have been extensively used in wireless systems along the time or frequency dimensions \cite{mo2016channel}. Recent work on omni-directional precoding has used ZC sequences in linear \cite{omni_pre_zc1} and planar phased arrays  \cite{omni_pre_zc2} to generate wide beams. There is limited work, however, on the use of ZC sequences in space for CS-based beam alignment. To the best of our knowledge, \cite{ZC_CS_first} is the only work that used ZC sequences in phased arrays for CS-based channel estimation. The optimality of ZC-based CS matrices in terms of incoherence or RIP, however, was not studied in \cite{ZC_CS_first}. Furthermore, \cite{ZC_CS_first} assumed perfect synchronization like most of the prior work on CS-based channel estimation. Random circulantly shifted ZC sequences were used in \cite{convCS} for compressed sensing of signals that are sparse in the DFT-basis. Our interpretation of shifted ZC-based CS for such signals, however, is different from that in \cite{convCS}. Specifically, we propose novel concepts of spectral mask, virtual switching and trajectory for shifted ZC-based CS. We also optimize the sequence of circulant shifts for robustness to phase errors due to CFO. 
\par We use $\mathbf{z}\in \mathbb{C}^{N}$ to denote the fixed ZC sequence, also defined as the core ZC sequence. The  $k^{\mathrm{th}}$ element of $\mathbf{z}$ is given as \cite{zadoffchu}
\begin{equation}
z\left[k\right]=\begin{cases}
\begin{array}{c}
\frac{1}{\sqrt{N}}\mathrm{exp}\left(\mathrm{j} \frac{\pi u k\left(k+1\right)}{N}\right),\,\,\,\,\,\mathrm{if}\,N\,\mathrm{is\,odd}\\
\frac{1}{\sqrt{N}}\mathrm{exp}\left(\mathrm{j} \frac{\pi u k^{2}}{N}\right),\,\,\,\,\,\,\,\,\,\,\,\,\,\,\,\,\mathrm{if}\,N\,\mathrm{is\,even}
\end{array}\end{cases}.
\end{equation} 
The integers $u$ and $N$ of the ZC sequence have to be co-prime for $\mathbf{z}$ to satisfy the Zero periodic Auto Correlation (ZAC) property \cite{zadoffchu}. The ZAC property translates to basis orthogonality of all the $N$ possible shifted ZC sequences. Furthermore, as the DFT of a ZC sequence is another ZC sequence \cite{ZCZC}, $\mathbf{z}$ has a uniform DFT in the magnitude sense. As every canonical basis vector also has a uniform DFT, shifted ZC sequences can be used to emulate the canonical basis set. We explain the canonical basis emulation effect of ZC sequences in Sec. \ref{sec:smvs}.
\subsection{Spectral mask and virtual switching} \label{sec:smvs}
In this section, we show how the phased array can be interpreted as a virtual switching architecture when ZC sequences are used along the spatial dimension. Consider a phase shift matrix $\mathbf{b}_n\mathbf{d}^T_n$, such that $\mathbf{b}_n$ and $\mathbf{d}_n$ are generated using $r[n] $ and $c[n]$ right circulant shifts of the core ZC sequence $\mathbf{z}$. The integers  $r [n]\in \mathcal{I}_N$ and $c[n]\in \mathcal{I}_N$ control the phase of the DFTs of $\mathbf{b}_n$ and $\mathbf{d}_n$. We use $\mathbf{\Lambda}_\mathbf{z}$ to denote a diagonal matrix containing the scaled DFT of $\mathbf{z}$ on its diagonal, i.e., $\mathbf{\Lambda}_\mathbf{z}=\sqrt{N}\mathrm{diag}\left(\mathbf{U}_N\mathbf{z}\right)$. Notice that each diagonal entry of $\mathbf{\Lambda}_{\mathbf{z}}$ has unit magnitude by the uniform DFT property of the ZC sequence. We define a right circulant delay matrix $\mathbf{J}\in \mathbb{R}^{N \times N}$ such that its first row is $\left(0,1,0,...,0\right)$. The vector $\mathbf{Jz}$ is then $(z[2],z[3],....,z[N],z[1])^{T}$. We define an $N \times N$ matrix containing all possible circulant shifts of $\mathbf{z}$ as $\mathbf{Z}=\left[\mathbf{z},\mathbf{Jz},\mathbf{J}^{2}\mathbf{z},...,\mathbf{J}^{N-1}\mathbf{z}\right]$. As $\mathbf{Z}$ is circulant, it can be diagonalized by the DFT as $\mathbf{U}_N \mathbf{Z} \mathbf{U}^{\ast}_N=\mathbf{\Lambda}_{\mathbf{z}}$ or $\mathbf{U}_N \mathbf{Z}=\mathbf{\Lambda}_{\mathbf{z}}\mathbf{U}_N$. From the structure of $\mathbf{Z}$, we have $\mathbf{b}_n=\mathbf{Z}\mathbf{e}_{r[n]}$. Therefore, the DFT of $\mathbf{b}_n$ is
\begin{equation}
\mathbf{U}_N \mathbf{b}_n= \mathbf{\Lambda}_{\mathbf{z}}\mathbf{U}_N \mathbf{e}_{r[n]}.
\end{equation}
Similarly, we have $\mathbf{U}_{N}\mathbf{d}_{n}=\mathbf{\Lambda_{\mathbf{z}}}\mathbf{U}_{N}\mathbf{e}_{c[n]}$. For $\epsilon=0$, we rewrite \eqref{eq:meas} using the 2D-DFT representation of $\mathbf{H}$ in \eqref{eq:mimoangledom} as 
\begin{equation}
y[n]=\mathbf{b}_{n}^{\ast}\mathbf{U}^{\ast}_{N}\mathbf{X}\mathbf{U}^{\ast}_{N}\overline{\mathbf{d}}_{n}+v[n].
\label{eq:normalcs}
\end{equation}
When $\mathbf{b}_n$ and $\mathbf{d}_n$ are $r[n]$ and $c[n]$ circulantly shifted versions of $\mathbf{z}$, the channel measurements in \eqref{eq:normalcs} can be expressed as 
\begin{align}
y[n]&=\left(\mathbf{U}_{N}\mathbf{b}_{n}\right)^{\ast}\mathbf{X} \overline{\mathbf{U}_{N}\mathbf{d}_{n}}+v[n]\\
&=\mathbf{e}_{r[n]}^{T}\mathbf{U}_{N}^{\ast}\mathbf{\Lambda_{\mathbf{z}}^{\ast}}\mathbf{X}\mathbf{\Lambda_{\mathbf{z}}^{\ast}}\mathbf{U}_{N}^{\ast}\mathbf{e}_{c[n]}+v[n].
\label{eq:transzcsamp}
\end{align}
We define $\mathbf{S}=\mathbf{\Lambda^{\ast}_{\mathbf{z}}}\mathbf{X}\mathbf{\Lambda^{\ast}_{\mathbf{z}}}$ as the spectral masked version of the true beamspace matrix $\mathbf{X}$. With this definition, it can be observed from \eqref{eq:transzcsamp} that the noiseless $y[n]$ is a $(r[n],c[n])$ coordinate of the inverse 2D-DFT of $\mathbf{S}$. As $\mathbf{\Lambda}_{\mathbf{z}}$ is diagonal, $\mathbf{S}$ is also sparse with the same locations of sparsity as that of $\mathbf{X}$. Because all the diagonal elements of $\mathbf{\Lambda_{\mathbf{z}}}$ have unit magnitude, the transformation from $\mathbf{X}$ to $\mathbf{S}$ is perfectly conditioned and invertible. As a result, estimating the masked beamspace matrix $\mathbf{S}$ is equivalent to estimating the true beamspace matrix $\mathbf{X}$. This equivalence allows us to define the spectral mask concept, i.e., the channel measurements acquired with shifted ZC sequences can be interpreted as partial 2D-DFT samples of a masked beamspace. 
\par The virtual switching concept defines a virtual channel whose entries can be directly sampled using ZC-based projections of the true channel. Analogous to the definition of $\mathbf{H}$ in \eqref{eq:mimoangledom}, we define $\mathbf{G}=\mathbf{U}^{\ast}_N \mathbf{S} \mathbf{U}^{\ast}_N$ as the virtual channel matrix. Using these definitions, \eqref{eq:transzcsamp} is then  
\begin{align}
y[n]&=\mathbf{e}_{r[n]}^{T}\mathbf{U}_{N}^{\ast}\mathbf{S}\mathbf{U}_{N}^{\ast}\mathbf{e}_{c[n]}+v[n] \label{eq:pretranszc}\\
&=\mathbf{e}_{r[n]}^{T} \mathbf{G} \mathbf{e}_{c[n]} +v[n] \label{eq:pretranszc2}\\
&=\mathbf{G}\left(r[n],c[n]\right)+v[n].
\label{eq:transformzc}
\end{align}
On the one hand, antenna switching cannot be performed in the phased array architecture, i.e., the entries of $\mathbf{H}$ cannot be sampled directly. On the other hand, it follows from \eqref{eq:pretranszc} that the entries of the virtual channel matrix can be sampled directly.  Such sampling is achieved by applying the shifted ZC-based training to the phased array. As the map between $\mathbf{G}$ and $\mathbf{H}$ is unique and sparsity preserving, it can be concluded from \eqref{eq:normalcs} and \eqref{eq:transformzc} that circulantly shifted ZC sequences emulate the canonical basis. The basis emulation can also be interpreted as a virtual transform that maps the phased array architecture to a switching-based architecture. The circulant shifts of the ZC sequence along the elevation and the azimuth, i.e., $\left(r[n],c[n]\right)$, would represent the 2D coordinates of the sampled antenna in the virtual switching architecture.  
\par Now, we show how virtual switching is advantageous over physical antenna switching when there is a limit on the maximum transmit power per-antenna. To explain our argument, we consider $\mathbf{H}=\mathbf{1}\mathbf{1}^T$ as the all-ones channel matrix. We also assume a maximum transmit power of $1/N$ units per-antenna. The $(0,0)$ coordinate of the true channel, i.e., $\mathbf{H}(0,0)$, can be acquired by using $\mathbf{b}_n=[1/ \sqrt{N},0,0,...,0 ]$ and $\mathbf{d}_n=[1/ \sqrt{N},0,0,...,0]$. For this choice of $\mathbf{b}_n$ and $\mathbf{d}_n$, 
it can be observed from \eqref{eq:normalcs} that the channel measurement is distributed as $\mathcal{CN}(1/N, \sigma^2 )$ and its associated SNR is $\mathrm{SNR}_{\mathrm{sw}}= 1/(N \sigma)^2$. To sample $\mathbf{G}(0,0)$, the $(0,0)$ coordinate of the virtual channel, $\mathbf{b}_n=\mathbf{z}$ and $\mathbf{d}_n=\mathbf{z}$ are applied to the phased array. The projection of the channel onto $\mathbf{b}_n\mathbf{d}^T_n$, i.e., $\mathbf{b}^{\ast}_n \mathbf{H}\overline{\mathbf{d}}_n$, is $|\mathbf{z}^{\ast}\mathbf{1}|^{2}$ when $\mathbf{H}$ is an all-ones matrix. Using the uniform DFT-property of the ZC sequence \cite{ZCZC}, it can be shown that $|\mathbf{z}^{\ast}\mathbf{1}|^2$, the scaled energy in the DC-component of $\mathbf{z}$, is $1$. As a result, the channel measurement corresponding to the virtual channel entry is distributed as $\mathcal{CN}(1,\sigma^2)$ and its associated SNR is $\mathrm{SNR}_{\mathrm{vsw}}= 1/ \sigma^2$. The SNR of the channel measurements obtained using ZC sequences is higher when compared to antenna switching because ZC sequences use all the antennas unlike switching. As $\mathrm{SNR}_{\mathrm{vsw}}=N^2 \mathrm{SNR}_{\mathrm{sw}}$, it can be concluded that virtual switching is power efficient over physical antenna switching by a factor of $N^2$ when there is a power constraint per-antenna.
\subsection{Analysis of CS algorithms with the proposed ZC-based training} 
The spectral mask concept aids fast CS algorithms and allows analysis of CS-based channel estimation. Let $M$ be the number of channel measurements obtained using distinct circulant shifts of the ZC sequence along the two dimensions of the UPA. Due to the sparse nature of $\mathbf{X}$, it is possible to reconstruct $\mathbf{H}$ with $M \ll N^2$ channel measurements. With the virtual switching interpretation, a simple CS acquisition strategy is to sample $M$ coordinates of $\mathbf{G}$ uniformly at random from an $N\times N$ grid without repetition. Let $\Omega$ denote the set of $M$ sampled 2D coordinates $\left\{ \left(r[n],c[n]\right) \right\} _{n=0}^{M-1}$, with $r[n], c[n] \in \mathcal{I}_N$. We use $\mathcal{P}_\Omega : \mathbb{C}^{N\times N}\rightarrow\mathbb{C}^{M}$ to denote the sampling operator that returns the entries of an $N\times N$ matrix at the locations in $\Omega$. For the $n^{\mathrm{th}}$ beam training measurement, the TX applies $\mathbf{b}_n\mathbf{d}^T_n$  to the UPA, where $\mathbf{b}_n$ and $\mathbf{d}_n$ are $r[n]$ and $c[n]$ circulantly shifted versions of the core ZC sequence $\mathbf{z}$. Let $\mathbf{y}=\left\{ y[n]\right\} _{n=0}^{M-1}$ denote the vector of channel measurements obtained by switching according to $\Omega$, i.e., $\mathbf{y}=\mathcal{P}_{\Omega}\left(\mathbf{G}\right)+\mathbf{v}$. Using \eqref{eq:pretranszc} and \eqref{eq:transformzc}, $\mathbf{y}$ can be interpreted as a sub-sampled 2D-DFT of the sparse masked beamspace matrix $\mathbf{S}$. The spectral mask concept thereby allows the application of CS guarantees for sub-sampled 2D-DFT matrices \cite{DFTCS} to the beam-alignment problem.
\par Now, we provide reconstruction guarantees for sparse channel estimation when random circulantly shifted ZC sequences are used for beam training. The masked beamspace $\mathbf{S}$ can be estimated from the channel measurements using the following convex program \cite{csintro}
\begin{equation}
\hat{\mathbf{S}}=\begin{array}{c}
\mathrm{arg\,min\,\,}\left\Vert \mathbf{W}\right\Vert _{1}\\
\mathrm{s.t\,}\left\Vert \mathbf{y}-\mathcal{P}_{\Omega}\left(\mathbf{U}_{N}^{\ast}\mathbf{W}\mathbf{U}_{N}^{\ast}\right)\right\Vert _{2}\leq\sqrt{M}\sigma
\end{array}.
\label{eq:solvel1}
\end{equation}
The objective of \eqref{eq:solvel1}, i.e., $\left\Vert \mathbf{W} \right\Vert _{1}=\sum_{k,\ell}\left|\mathbf{W}\left(k,\ell\right)\right|$, is a convex function that encourages sparse solutions. We substitute the parameters of the sparse masked beamspace recovery problem in \eqref{eq:solvel1}, in Theorems 1 and 3 of \cite{kramer} to obtain recovery guarantees on the channel estimate. The matrix $\left( \mathbf{S}\right)_k$ is used to denote the $k$ sparse representation of $\mathbf{S}$, and is obtained from $\mathbf{S}$ by retaining the $k$ largest entries in magnitude and setting the rest to $0$.
\begin{theorem} 
For a fixed constant $\gamma \in \left(0,1\right)$, if $M\geq Ck\,\mathrm{max}\left\{ 2\mathrm{log}^{3}(2k)\,\mathrm{log}(N),\,\mathrm{log}(\gamma^{-1})\right\}$ then the solution to \eqref{eq:solvel1} satisfies 
\begin{equation}
\bigl\Vert \mathbf{S}-\hat{\mathbf{S}}\bigl\Vert _{F}\leq C_{1}\frac{\left\Vert \mathbf{S}-\left(\mathbf{S}\right)_{k}\right\Vert _{1}}{\sqrt{k}}+C_{2}N\sigma,
\label{eq:CSbound}
\end{equation}
with a probability of at least $1-\gamma$. The constants $C, C_1$ and $C_2$ are independent of all the other parameters.
\label{theoremorig}
\end{theorem}
The beam space estimate $\hat{\mathbf{X}}$ and the corresponding channel estimate $\hat{\mathbf{H}}$ can be computed as $\hat{\mathbf{X}}=\overline{\mathbf{\Lambda}}^{-1}_{\mathbf{z}}\hat{\mathbf{S}}\overline{\mathbf{\Lambda}}^{-1}_{\mathbf{z}}$ and $\hat{\mathbf{H}}=\mathbf{U}^{\ast}_N \hat{\mathbf{X}}\mathbf{U}^{\ast}_N$. Due to the norm preserving nature of the phase map from $\mathbf{S}$ to $\mathbf{X}$, we have $\bigl\Vert \mathbf{X}-\hat{\mathbf{X}}\bigl\Vert _{F}=\bigl\Vert \mathbf{S}-\hat{\mathbf{S}}\bigl\Vert _{F}$ and $\bigl\Vert \mathbf{X}-(\mathbf{X})_k \bigl\Vert_{1}=\bigl\Vert \mathbf{S}-(\mathbf{S})_k\bigl\Vert _{1}$. Furthermore, as $\mathbf{H}$ and $\mathbf{X}$ are related through a unitary 2D-DFT, we have  $\bigl\Vert \mathbf{H}-\hat{\mathbf{H}}\bigl\Vert _{F}=\bigl\Vert \mathbf{X}-\hat{\mathbf{X}}\bigl\Vert _{F}$. Putting these observations together, the error in the estimated channel can be bounded as 
\begin{equation}
\bigl\Vert \mathbf{H}-\hat{\mathbf{H}}\bigl\Vert _{F}\leq C_{1}\frac{\bigl\Vert \mathbf{X}-\left(\mathbf{X}\right)_{k}\bigl\Vert _{1}}{\sqrt{k}}+C_{2}N\sigma 
\label{eq:csfinal}
\end{equation}
for the conditions in Theorem \ref{theoremorig}. As seen in \eqref{eq:csfinal}, the CS algorithm in \eqref{eq:solvel1} can be used to recover a sparse approximation of the beamspace channel ${\mathbf{X}}$ with just $\mathcal{O}(\mathrm{log}N)$ channel measurements using the proposed shifted ZC-based training.  
\par The proposed circulantly shifted ZC sequences are used differently from those in \cite{mo2016channel} where the shifts were performed along the time dimension to reduce the complexity of the channel estimation algorithm. Our work exploits circulant shifts along the spatial dimension to design phased array compatible training sequences that minimize the number of CS measurements, allow recovery guarantees, and also aid fast implementation. Such guarantees and complexity reduction are otherwise difficult to achieve for the commonly used random IID phase shift sequences.
\section{Impact of CFO on beam alignment} \label{sec:CFOimpactfull}
\par In this work, we study the impact of CFO on compressive beam training and design beam training sequences that achieve robustness to CFO. For tractability, we restrict our design within the class of sequences generated using circulant shifts of the ZC sequence $\mathbf{z}$ along the azimuth and elevation dimensions. It can be noticed from \eqref{eq:transformzc} that the circulant shift determines the coordinate of the virtual channel matrix to be sampled. For a sequence of $M$ beam training measurements obtained using circulant shifts $\left\{ (r[n],c[n])\right\} _{n=0}^{M-1}$, we define a sampling trajectory as a curve that sequentially traverses through the coordinates $\left\{ (r[n],c[n])\right\} _{n=0}^{M-1}$ on an $N \times N$ grid. Some examples of trajectories for $M=N^2$ are shown in Fig. \ref{fig:rowtraj} and Fig. \ref{fig:blktraj}. We show that the impact of CFO on the beamspace estimate is determined by the coordinates sampled by the trajectory and the sequence of traversal. Therefore, training design for low complexity beam training robust to CFO is now a trajectory design problem.
\par In this section, we provide examples of some full-sampling trajectories and analyze the beamspace distortions induced by CFO. Using the insights derived from these examples, we arrive at design requirements for a good trajectory. 
\subsection{Impact of trajectory on beamspace estimate} \label{trajectory_effects}
For a non-zero CFO, the measurements in \eqref{eq:transformzc} can be expressed in terms of the virtual channel $\mathbf{G}$ as
\begin{align}
y[n]&=e^{\mathrm{j} \epsilon n} \mathbf{G}\left(r[n],c[n]\right)+v[n].
\label{eq:canonical}
\end{align}
Irrespective of the choice of the trajectory, the phase errors in the measurements \eqref{eq:canonical} increase linearly along the sampling trajectory. In this section, we consider full-sampling trajectories  that traverse through all the $N^2$ points on the $N \times N$ grid exactly once, i.e., $M=N^2$. 
When CFO is ignored, the virtual channel matrix $\mathbf{G}$ can be estimated by sequentially populating an $N \times N$ matrix with the entries of $\mathbf{y}$ along the trajectory. In such reconstruction, the entries of the estimated $\mathbf{G}$ differ from that of the true $\mathbf{G}$ by a phase error that depends on CFO and the trajectory. For a given trajectory, we use $\mathbf{P}_{\epsilon} \in \mathbb{C}^{N\times N}$ to denote the matrix containing the phase errors induced by CFO. As CFO induces a phase error of $e^{\mathsf{j}m\epsilon}$ in the $m^{\mathrm{th}}$ channel measurement, the entries of $\mathbf{P}_{\epsilon}$ are of the form $e^{ \mathrm{j} m\epsilon}$, where $m \in \mathcal{I}_{N^2}$. Ignoring CFO, the noiseless virtual channel estimate is then $\hat{\mathbf{G}}=\mathbf{G} \odot \mathbf{P}_{\epsilon}$. The corresponding masked beamspace is estimated as $\hat{\mathbf{S}}=\mathbf{U}_N \hat{\mathbf{G}} \mathbf{U}_N$. Using the property that element-wise multiplication of two matrices results in 2D circular convolution of their 2D-DFTs \cite{imageprocess}, we have  $\hat{\mathbf{S}}=\mathbf{S} \circledast \left( \mathbf{U}_N \mathbf{P}_{\epsilon} \mathbf{U}_N \right)$. Therefore, the impact of CFO is determined by the matrix $\mathbf{U}_N \mathbf{P}_{\epsilon} \mathbf{U}_N$, which must be $N \mathbf{e}_{0}\mathbf{e}_{0}^{T}$ in the ideal case, i.e., for zero distortion in the masked beamspace. The original beamspace matrix is estimated by inverting the spectral mask as $\hat{\mathbf{X}}=\overline{\mathbf{\Lambda}}^{-1}_{\mathbf{z}}\hat{\mathbf{S}}\overline{\mathbf{\Lambda}}^{-1}_{\mathbf{z}}$. Using the examples of row sampling and block sampling, we show that CFO induces undesirable beamspace distortions like global shifts and multiple replicas that impact beam alignment. 
\subsubsection{Row sampling}
\begin{figure}[htbp]
\vspace{-5mm}
\subfloat[Row sampling trajectory]{\includegraphics[width=4.25cm, height=4.25cm]{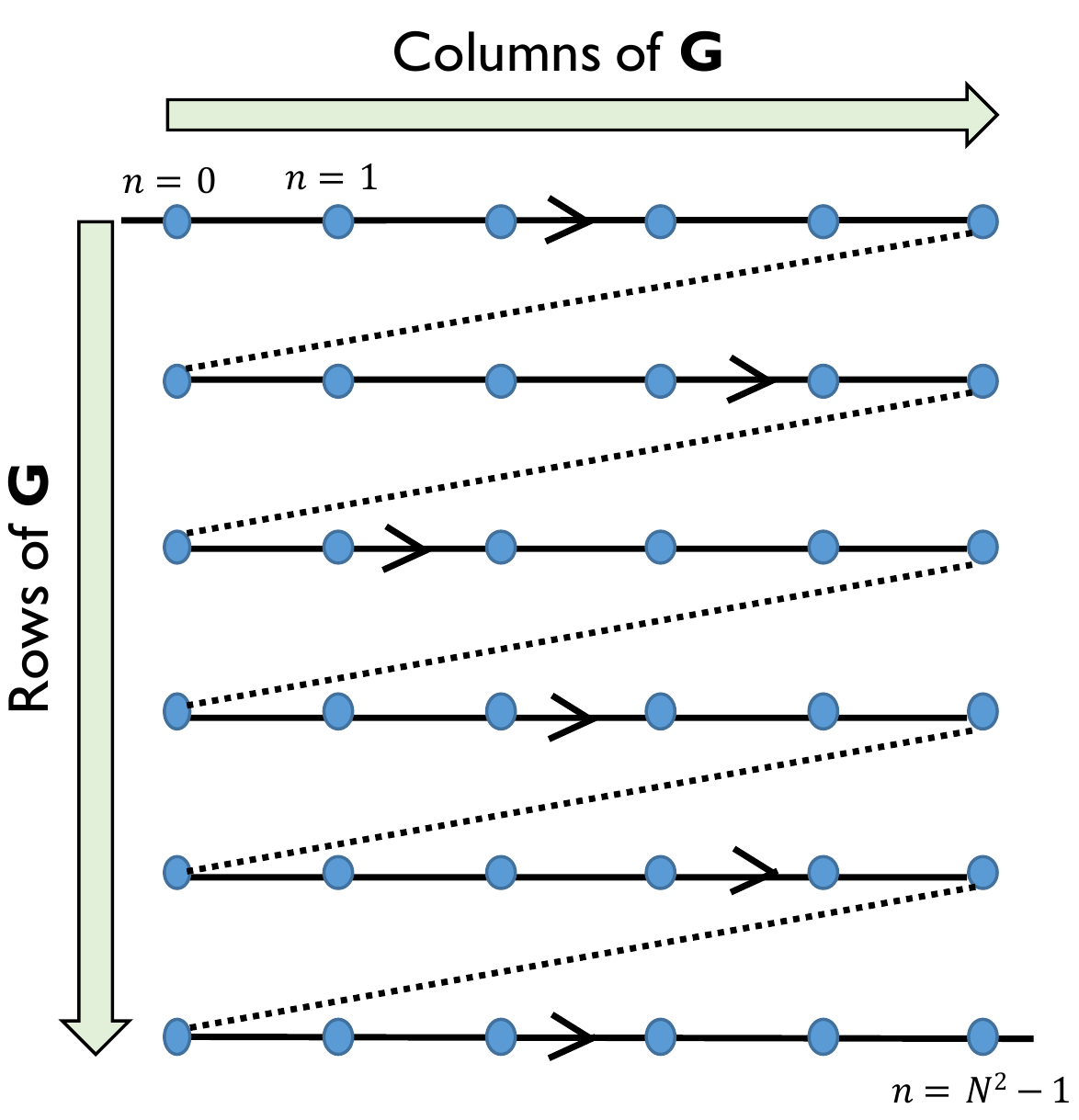}\label{fig:rowtraj}}
\hspace*{\fill}
\subfloat[Estimated beamspace matrix]{\includegraphics[trim=1cm 0cm 1cm 1cm,clip=true,width=4.25cm, height=4.25cm]{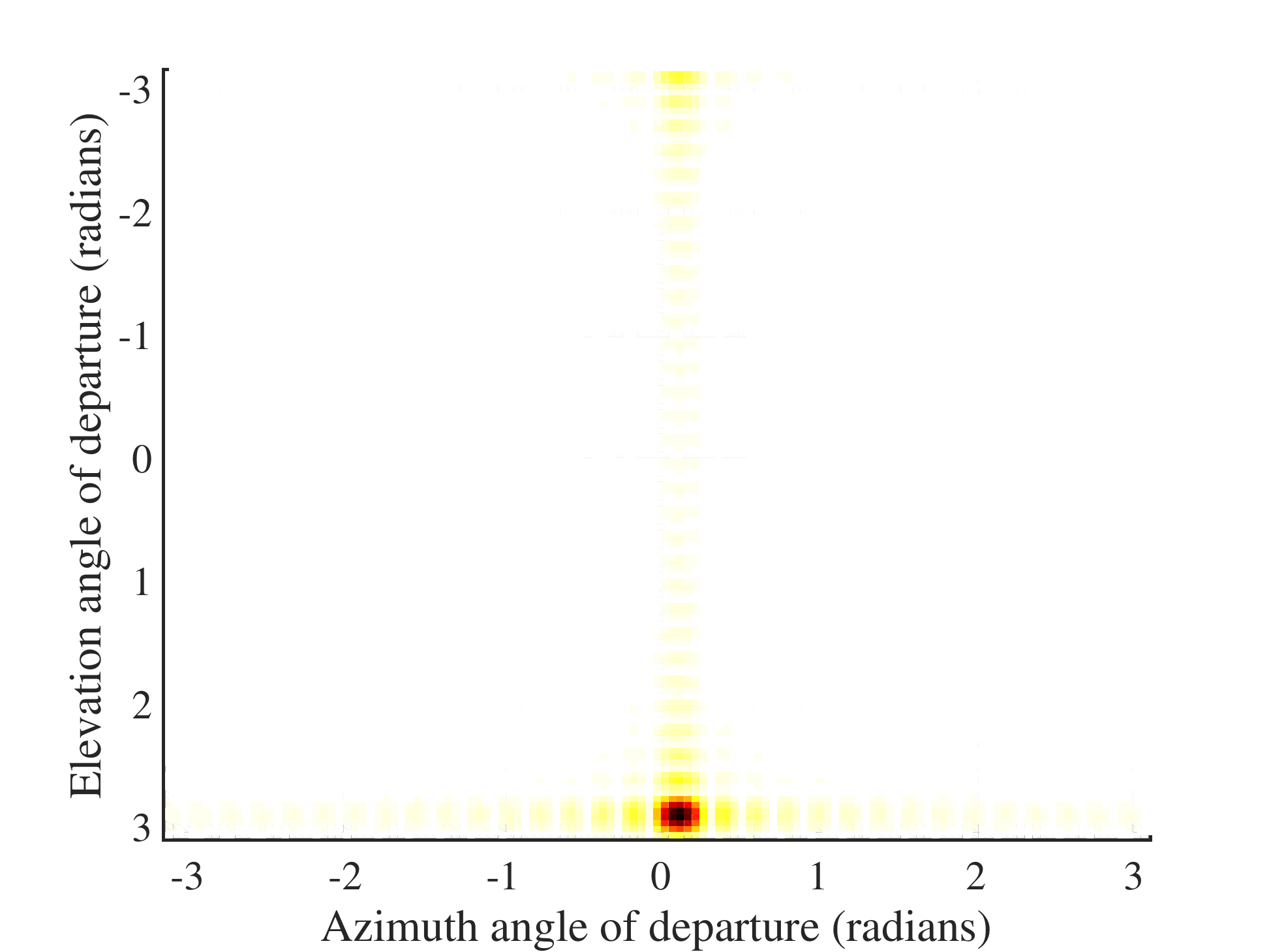}}\label{fig:rowrec}
\caption{ \small CFO induces beamspace shift when the row sampling trajectory is used for beam training. In this example, the original beamspace had a unique non-zero component at $(0,0)$, and $\epsilon$ was $0.09\,\mathrm{rad}$. \normalsize}
\vspace{-2mm}
\end{figure}
The row sampling trajectory, shown in Fig.~\ref{fig:rowtraj}, starts at $\left(0,0 \right)$ and sequentially samples every row of the $N\times N$ virtual channel matrix. To realize such a trajectory with the UPA, $\mathbf{z}$ is circulantly shifted by $r[n]=\lfloor n / N \rfloor$ and $c[n]= n \, \mathrm{modulo}\, N$ along the elevation and the azimuth dimensions. Then, the outer product of the two circulantly shifted vectors is applied to the UPA at the TX to obtain the channel measurement $y[n]$. For the row sampling trajectory, we use $\mathbf{P}^{\mathrm{row}}_{\epsilon} \in \mathbb{C}^{N \times N}$ to represent the matrix containing the phase errors induced by the CFO. When $\mathbf{G}$ is sampled according to the trajectory in  Fig.~\ref{fig:rowtraj}, the induced phase error matrix can be constructed by sequentially populating an $N \times N$ matrix with the elements $\{e^{\mathsf{j}m \epsilon}\}_{m \in \mathcal{I}_{N^2}}$ along the trajectory. In such case, the phase errors induced by CFO increase linearly along each row of $\mathbf{G}$. For the row sampling trajectory, the induced phase error matrix $\mathbf{P}^{\mathrm{row}}_{\epsilon}$ in compact form is $\mathbf{P}^{\mathrm{row}}_{\epsilon}=\mathbf{a}_N\left(N \epsilon \right)\mathbf{a}_N\left(\epsilon \right)^T$, 
and its 2D Fourier transform is concentrated at the frequency coordinate $\left(N \epsilon, \epsilon \right)$. In this case, the estimated virtual channel matrix is a noisy version of $\mathbf{G} \odot \mathbf{P}^{\mathrm{row}}_{\epsilon}$ for a CFO of $\epsilon$. Using the multiplication-convolution duality of the Fourier transform \cite{OPP}, it can be observed that CFO induces a 2D circulant shift of $\left(N \epsilon, \epsilon \right)$ in the masked beamspace. Specifically, $\hat{\mathbf{S}}$ is obtained by circulantly shifting the continuous 2D Fourier transform \cite{imageprocess} of $\mathbf{G}$ by $\left(N \epsilon, \epsilon \right)$ followed by sampling over the $N^2$ frequencies $(2\pi k/N,2\pi\ell/N)$ for $k, \ell \in \mathcal{I}_N$. By inverting the spectral mask, it can be observed that the estimated beamspace and the true beamspace differ by a 2D circulant shift of $(N\epsilon, \epsilon)$.
\par A key observation is that the beamspace shift induced by the row sampling trajectory is linear in the CFO and is higher along the elevation dimension. As the angular resolution of the array is $2 \pi / N$ along each dimension of the beamspace, beamforming with uncorrected beamspace shift can lead to significant loss in the SNR unless $\mathrm{max}\left\{N\epsilon, \epsilon \right\} < \pi / N$. The condition $\epsilon < \pi / N^2$ results in a strict limit on CFO and may not be satisfied in mmWave systems as the number of antennas, i.e., $N^2$, is typically large. Therefore, it is desirable  to design trajectories that minimize the beamspace shifts along both the azimuth and the elevation dimensions. 
\vspace{-1mm}
\subsubsection{Block sampling}
\begin{figure}[htbp]
\subfloat[Block sampling trajectory]{\includegraphics[width=4.25cm, height=4.25cm]{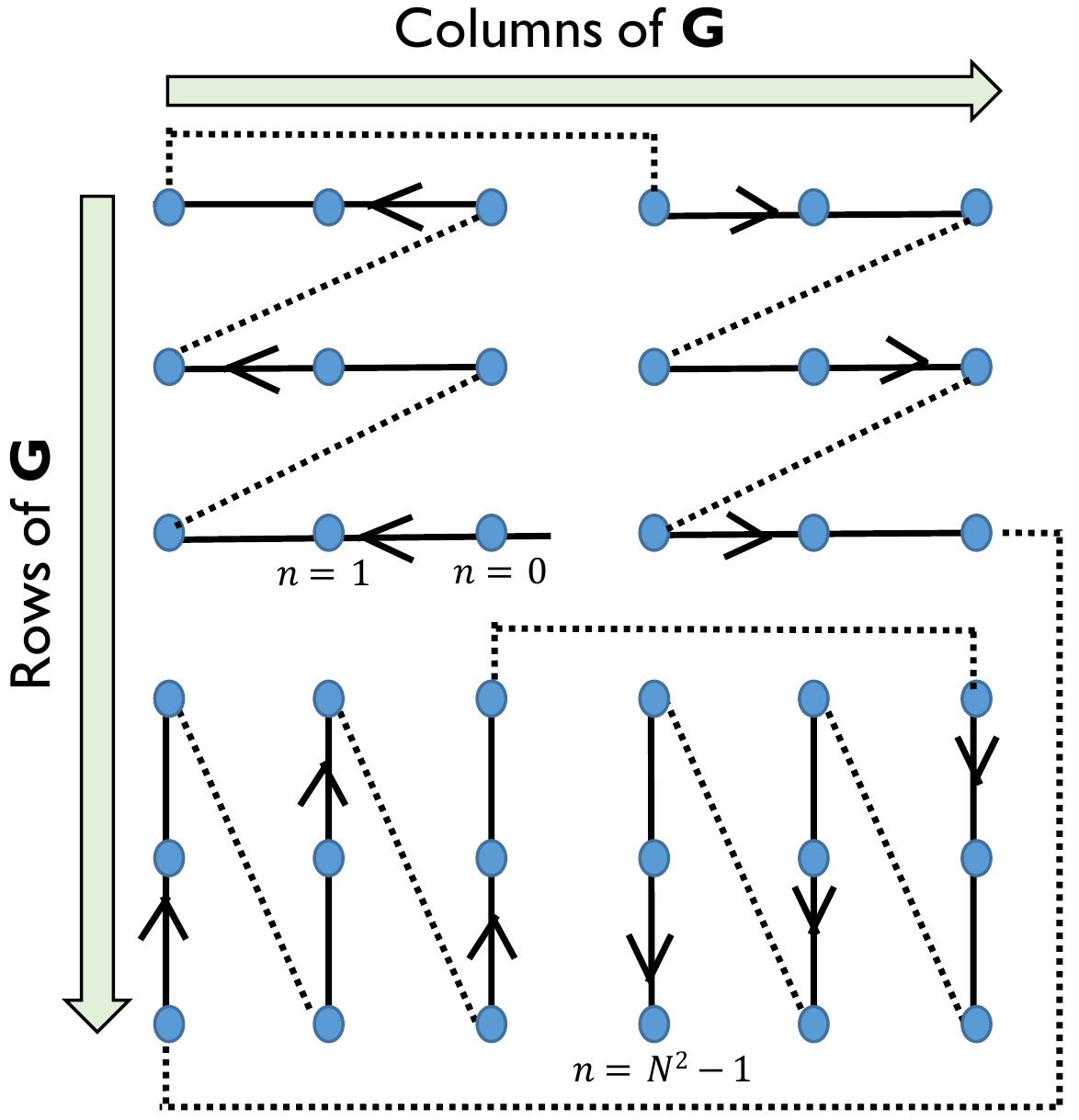}\label{fig:blktraj}} 
\hspace*{\fill}
\subfloat[Estimated beamspace matrix]{\includegraphics[trim=1cm 0cm 1cm 1cm,clip=true,width=4.25cm, height=4.25cm]{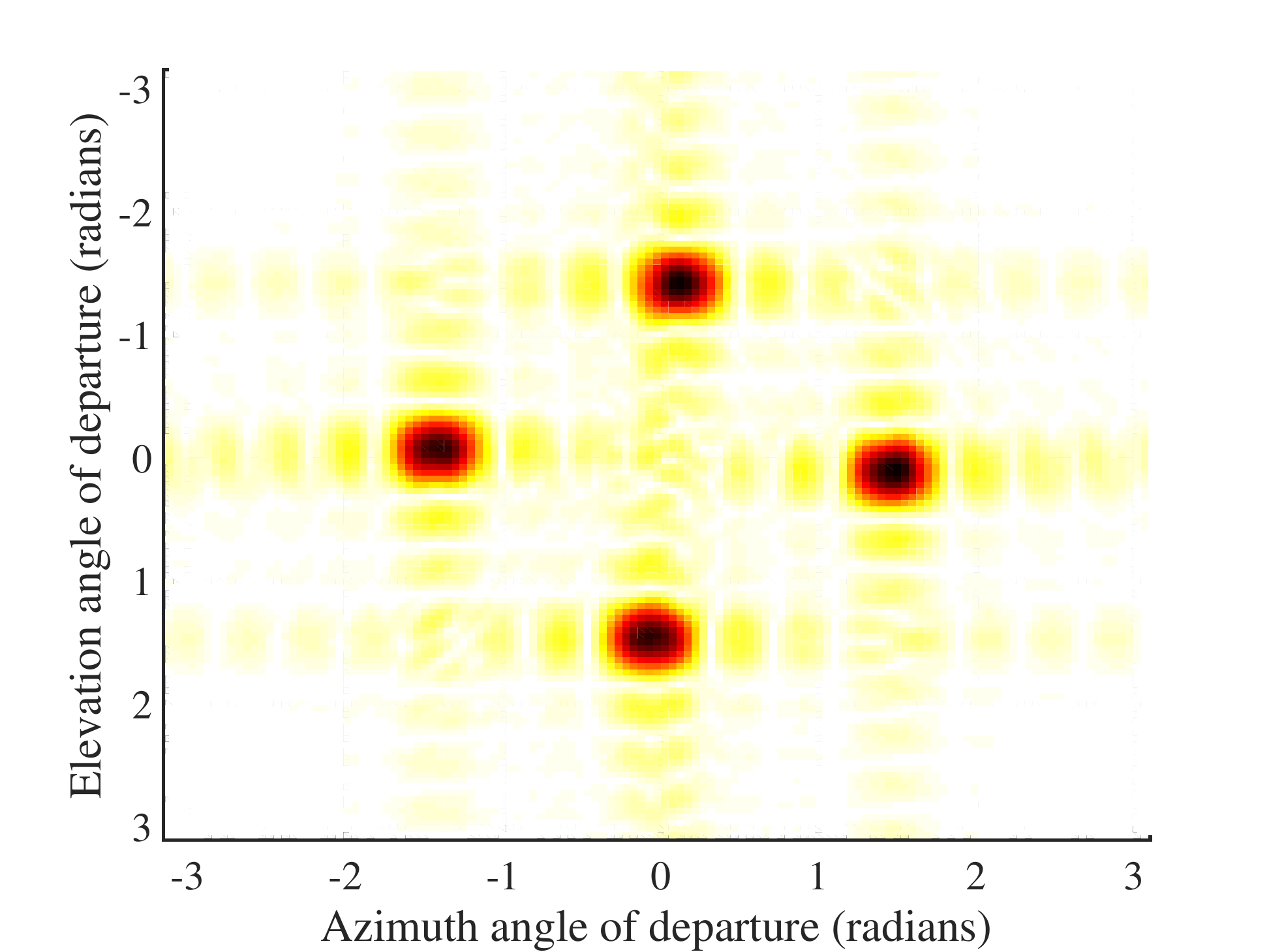}\label{fig:blkrec}}
\caption{\small CFO results in multiple replicas in the estimated beamspace with the block sampling trajectory. In this example, the original beamspace had a unique component at $(0,0)$, and $\epsilon$ was $0.09 \,\mathrm{rad}$.
\normalsize}
\vspace{-2mm}
\end{figure}
We show that CFO can induce multiple replicas in the beamspace estimate using a block sampling trajectory with $M=N^2$. From the structure of the sampling trajectory in Fig.~\ref{fig:blktraj}, the induced phase error matrix can be computed as 
\small
\begin{equation}
\mathbf{P}^{\mathrm{blk}}_{\epsilon}=\left[\begin{array}{cc}
e^{\mathrm{j} \epsilon u_{1}}\mathbf{a}_{\frac{N}{2}}\left(\frac{-N\epsilon}{2}\right)\mathbf{a}_{\frac{N}{2}}^{T}\left(-\epsilon\right) & e^{\mathrm{j} \epsilon u_{2}}\mathbf{a}_{\frac{N}{2}}\left(\frac{N\epsilon}{2}\right)\mathbf{a}_{\frac{N}{2}}^{T}\left(\epsilon\right)\\
e^{\mathrm{j} \epsilon u_{3}}\mathbf{a}_{\frac{N}{2}}\left(-\epsilon\right)\mathbf{a}_{\frac{N}{2}}^{T}\left(\frac{N\epsilon}{2}\right) & e^{\mathrm{j} \epsilon u_{4}}\mathbf{a}_{\frac{N}{2}}\left(\epsilon\right)\mathbf{a}_{\frac{N}{2}}^{T}\left(\frac{-N\epsilon}{2}\right)
\end{array}\right],
\label{eq:blocksamp}
\end{equation}
\normalsize
where $u_1={N^2} /{4}-1$, $u_2={N^2}/{4}$, $u_3={N^2}/{2}+{N}/{2}-1$ and $u_4=N^2-{N}/{2}$. It can be noticed that the phase errors in $\mathbf{P}^{\mathrm{blk}}_{\epsilon}$ increase linearly along the trajectory in Fig. \ref{fig:blktraj}. Each of the four $N/2 \times N/2 $ submatrices in  $\mathbf{P}^{\mathrm{blk}}_{\epsilon}$ have unique spectral components at $\left(-N\epsilon/2,-\epsilon\right),\left(N\epsilon/2,\epsilon\right),\left(-\epsilon,N\epsilon/2\right)$ and $\left(\epsilon,-N\epsilon/2\right)$ in the 2D Fourier representation. Therefore, the 2D Fourier transform of $\mathbf{P}^{\mathrm{blk}}_{\epsilon}$ has four distinct spectral components for an $\epsilon \neq 0$. For a CFO of $\epsilon$, the virtual channel matrix acquired with the block sampling trajectory is $\mathbf{G} \odot \mathbf{P}^{\mathrm{blk}}_{\epsilon}$. By the multiplication-convolution duality of the Fourier transform \cite{OPP}, we can conclude that CFO replicates each spectral component in $\mathbf{G}$ into four components in the masked beamspace, i.e., $\hat{\mathbf{S}}$. This replica effect, shown in Fig. \ref{fig:blkrec}, distorts the sparsity of the beamspace and is not desirable from a CS perspective. 
\par It can be observed that both the row and block sampling trajectories introduce ambiguity in recovering the true beamspace estimate. For example, it is difficult to infer if the beamspace matrix is the reconstructed matrix for an $\epsilon=0$, or the true beamspace matrix with the true $\epsilon$.  
\subsection{Are there trajectories that result in zero shift?} \label{sec:DCcomp}  
For the trajectory to result in zero beamspace shift, the 2D Fourier transform of $ \mathbf{P}_{\epsilon}$ must have  a peak at $\left(0,0\right)$ as $\hat{\mathbf{S}}=\mathbf{S} \circledast \left( \mathbf{U}_N \mathbf{P}_{\epsilon} \mathbf{U}_N \right)$. Furthermore, the spread about the peak must be minimum to reduce blurring in the reconstructed beamspace. Irrespective of the choice of full-sampling trajectory, the DC spatial frequency component of $\mathbf{P}_{\epsilon}$ is $\sum_{n=0}^{N^2-1} e^{\mathrm{j} \epsilon n}/ N^2$ , and is approximately zero for typical values of $\epsilon$ and $N$. Therefore, any trajectory results in non-zero shift in the beamspace estimate in typical mmWave systems. 
\par Beam alignment using the distorted beamspace directly translates to reduction in the received SNR. For instance, any unknown global shift in the beamspace estimate would result in TX beam misalignment. Beamforming using blurred beamspace matrices can result in beams with a higher beam-width and reduce the received SNR. The problem of finding shift and blur minimizing trajectories is complicated for the undersampled case $M \ll N^2$, and is investigated in Sec. \ref{sec:usamp}.
\section{Undersampling trajectories for beamforming: A space-time game} \label{sec:usamp}
In this section, we design a randomized undersampling trajectory in the virtual antenna domain that minimizes the impact of CFO on standard CS algorithms.
\par  For the case of full-sampling, the conditions for a good trajectory translate to requirements on the induced phase error matrix $\mathbf{P}_{\epsilon}$. From the discussion in Sec. \ref{sec:DCcomp}, it can be concluded that the ideal $\mathbf{P}_{\epsilon}$, with $\mathbf{U}_N \mathbf{P}_{\epsilon} \mathbf{U}_N = N \mathbf{e}_0 \mathbf{e}^T_0$, cannot be realized for $\epsilon \neq 0$. In other words, the spectrum of $\mathbf{P}_{\epsilon}$ cannot be concentrated at the DC frequency, i.e., $(0,0)$. As a compromise, we seek a design such that the spectrum of $\mathbf{P}_{\epsilon}$ is concentrated at a point that is proximal to the DC frequency. It can be seen that the entries of the matrix $\mathbf{P}_{\epsilon}$ are constrained to be distinct elements in the set $\{ 1,e^{\mathrm{j} \epsilon},e^{\mathrm{j} 2\epsilon},...,e^{\mathrm{j} (N^{2}-1)\epsilon}\}$ for full-sampling. Therefore, finding the best trajectory such that the corresponding phase error matrix satisfies spectral concentration and proximity properties requires optimization over permutations of $N^2$ elements. The trajectory optimization problem is even more complicated for the undersampled case as it also requires selecting the subset of coordinates to sample.
\par In this work, we propose a sub-optimal solution to the trajectory optimization problem by first constructing a phase error matrix that satisfies the concentration and proximity properties, and then fit an undersampling trajectory that is consistent with the matrix. With $x$ and $y$ used to denote the row and column coordinates in a $N \times N$ matrix, we define the $k^{\mathrm{th}}$  contour as $\left\{ \left(x,y\right):\,x+y=k\right\}$. For reference, the coordinate of the first entry in a $N \times N$ matrix is set to $\left(0,0\right)$. The  $x$ and $y$ coordinates increase along the row and column of a matrix. The contours for $N=6$ are shown as dotted lines in Fig.~\ref{fig:pseq}.
By definition, an $N \times N$ matrix would have $2N-1$  contours. For a CFO of $\epsilon$, a phase error matrix satisfying spectral concentration and proximity is 
\begin{equation}
\mathbf{P}^{\mathrm{cnt}}_{\epsilon}=\left[\begin{array}{ccccc}
1 & e^{\mathrm{j}\epsilon} & e^{\mathrm{j}2\epsilon} & \cdots & e^{\mathrm{j}(N-1)\epsilon}\\
e^{\mathrm{j}\epsilon} & e^{\mathrm{j}2\epsilon} &  &  & e^{\mathrm{j}N\epsilon}\\
e^{\mathrm{j}2\epsilon} & \vdots & \iddots & \iddots & e^{\mathrm{j}(N+1)\epsilon}\\
\vdots &  &  &  & \vdots\\
e^{\mathrm{j}\left(N-1\right)\epsilon} & e^{\mathrm{j}N\epsilon} & e^{\mathrm{j}(N+1)\epsilon} & \cdots & e^{\mathrm{j}2(N-1)\epsilon}
\end{array}\right],
\label{eq:crossmat}
\end{equation}
where the phase errors are invariant along any  contour.  As $\mathbf{P}^{\mathrm{cnt}}_{\epsilon}=\mathbf{a}_N\left( \epsilon \right)\mathbf{a}^T_N\left( \epsilon \right)$, its Fourier transform is concentrated at $\left(\epsilon, \epsilon \right)$. The matrix $\mathbf{P}^{\mathrm{cnt}}_{\epsilon}$ cannot be realized with any full-sampling trajectory as the phase errors repeat along each contour. It can be noticed from \eqref{eq:canonical} that a single entry of $\mathbf{G}$ can be sampled in a given beam training slot, and CFO induces a phase error of $e^{\mathrm{j} k\epsilon}$ in the sample acquired in the $k^{\mathrm{th}}$ slot. Realizing the phase errors in $\mathbf{P}^{\mathrm{cnt}}_{\epsilon}$ requires sampling multiple entries of $\mathbf{G}$ in the same slot, which is not feasible. For example, all the $N$ entries of $\mathbf{G}$ on the $(N-1)^{\mathrm{th}}$ contour must be sampled in parallel in the $N^{\mathrm{th}}$ slot to meet the uniform phase error constraint along the $(N-1)^{\mathrm{th}}$ contour of $\mathbf{P}^{\mathrm{cnt}}_{\epsilon}$. 
\par Now, we explain how to design undersampling trajectories that realize the same effect as $\mathbf{P}^{\mathrm{cnt}}_{\epsilon}$. For simplicity of exposition, we consider the virtual channel matrix $\mathbf{G}$ in \eqref{eq:canonical} to be an all-ones matrix. We aim to undersample $\mathbf{G}$ in such a way that the reconstructed beamspace has spectrum concentrated at $\left( \epsilon, \epsilon \right)$ for a CFO of $\epsilon$. As all the entries of $\mathbf{G}$ are equal to $1$, the undersampled matrix contains just the phase errors induced by CFO. Our objective is to make the undersampled phase error matrix consistent with $\mathbf{P}^{\mathrm{cnt}}_{\epsilon}$ to the best possible extent. As sampling multiple entries of $\mathbf{G}$ in a given slot is not feasible, the undersampling trajectory must sample exactly one coordinate from the $k^{\mathrm{th}}$ contour of $\mathbf{G}$ during the $k^{\mathrm{th}}$ beam training slot. There would be $N\left(\left(N-1\right)!\right)^{2}$ undersampling trajectories of length $2N-1$ that satisfy this requirement; one of them is shown in Fig.~\ref{fig:pseq}. Channel acquisition using such undersampling trajectory yields an undersampled version of $\mathbf{G} \odot \mathbf{P}_{\epsilon}^{\mathrm{cnt}}$ for a CFO of $\epsilon$. 
\par We investigate the possibility of recovering $\mathbf{G} \odot \mathbf{P}_{\epsilon}^{\mathrm{cnt}}$ from its undersampled version. For any matrix $\mathbf{G}$, the 2D-Fourier transform of $\mathbf{G} \odot \mathbf{P}_{\epsilon}^{\mathrm{cnt}}$ is an $(\epsilon, \epsilon)$ shifted version of the 2D-Fourier transform of $\mathbf{G}$. When $\epsilon$ is an integer multiple of $2 \pi / N$, the matrix $\mathbf{G} \odot \mathbf{P}_{\epsilon}^{\mathrm{cnt}}$ has a sparse 2D-DFT representation whenever the 2D-DFT of $\mathbf{G}$ is sparse. This conclusion follows from the fact that the 2D-DFTs of $\mathbf{G}$ and $\mathbf{G} \odot \mathbf{P}_{\epsilon}^{\mathrm{cnt}}$ differ by an integer shift that lies on the DFT grid, and the observation that a shifted version of a sparse matrix is also sparse. When the 2D-DFT of $\mathbf{G}$ is sparse and $\epsilon$ is not an integer multiple of $2 \pi / N$, the 2D-DFT of $\mathbf{G} \odot \mathbf{P}_{\epsilon}^{\mathrm{cnt}}$ can be considered as approximately sparse due to leakage effects that arise from the off-grid CFO. CS algorithms can exploit the approximate sparse nature of $\mathbf{G} \odot \mathbf{P}_{\epsilon}^{\mathrm{cnt}}$ in the 2D-DFT basis to reconstruct the matrix $\mathbf{G} \odot \mathbf{P}_{\epsilon}^{\mathrm{cnt}}$ from its undersampled version. It is not clear, however, if the reconstruction is unique as it depends on the trajectory. We aim to find the trajectories that aid CS algorithms in recovering the $(\epsilon,\epsilon)$ shifted version of the masked beamspace.
\subsection{Swift-Link's undersampling mechanism}
\par We now describe p-trajectory, a component of Swift-Link's randomized trajectory, that is consistent with $\mathbf{P}^{\mathrm{cnt}}_{\epsilon}$. The trajectory is called as p-trajectory because a CFO of $\epsilon$ induces a positive shift of $(\epsilon,\epsilon)$ in the estimated beamspace. For the $n^{\mathrm{th}}$ beam training measurement, Swift-Link's p-trajectory samples a single coordinate of $\mathbf{G}$ at random from the coordinates on its $n^{\mathrm{th}}$  contour. In this work, we use uniform or binomial distributions to sample coordinates on each contour, and the coordinates are sampled independently across different contours. For instance, the binomial distribution based p-trajectory  samples one of the $3$ coordinates of the second contour according to the density $\left(0.25,0.5,0.25\right)$, while the uniform one samples according to $\left(1/3,1/3,1/3\right)$. A realization of the p-trajectory when the coordinates on each contour are sampled according to a uniform distribution is shown in Fig.~\ref{fig:pseq}. Standard CS using the samples acquired by the p-trajectory would result in an $(\epsilon, \epsilon)$ circulantly shifted version of the masked beamspace when successful. Therefore, standard CS-based beam alignment using the p-trajectory does not result in significant beam misalignment unless $\epsilon > \pi / N$. For a moderate CFO, beam-broadening \cite{Beambroad} can ensure a reasonable beamforming gain. To achieve sufficient link margin even for a high CFO, algorithms that correct the beamspace shifts must be developed.
\section{Swift-Link's strategy for shift correction}\label{sec:strategy}
The notion behind Swift-Link's low complexity algorithm for shift correction comes from the observation that traversing in the opposite direction of the p-trajectory results in circulant beamspace shifts in the opposite direction. From the two different circulantly shifted versions of the same masked beamspace matrix, it is possible to correct the shift and perform beamforming. We explain these ideas in detail to propose SwiftLink's type I and type II trajectories.
\par In this section, we derive a phase error matrix that satisfies the concentration and proximity properties, and helps in correcting the beamspace shift due to $\mathbf{P}_{\epsilon}^{\mathrm{cnt}}$. A phase error matrix that corresponds to a beamspace shift of $\left(-\epsilon, -\epsilon \right)$ for a CFO of $\epsilon$ is $\mathbf{P}^{\mathrm{cnt}}_{-\epsilon}$. The virtual channel acquired under such a phase error matrix is $\mathbf{G} \odot \mathbf{P}^{\mathrm{cnt}}_{-\epsilon}$. It can be observed from \eqref{eq:crossmat} that $\mathbf{P}^{\mathrm{cnt}}_{-\epsilon}= \mathbf{a}_N(-\epsilon)\mathbf{a}^T_N(-\epsilon)$, and the Fourier transform of $\mathbf{P}^{\mathrm{cnt}}_{-\epsilon}$ is concentrated at $(-\epsilon,-\epsilon)$. From the multiplication-convolution duality of the Fourier transform, it follows that the beamspace matrix corresponding to $\mathbf{G} \odot \mathbf{P}^{\mathrm{cnt}}_{-\epsilon}$ is a $(-\epsilon,-\epsilon)$ shifted version of the true beamspace. Unfortunately, the phase errors in $\mathbf{P}^{\mathrm{cnt}}_{-\epsilon}$ are negative integer multiples of $\epsilon$ and cannot be directly realized in practice. It is possible, however, to realize the same negative beamspace shift effect using the matrix $\mathbf{Q}_{\epsilon}^{\mathrm{cnt}}=e^{\mathrm{j}2 \left(N-1\right) \epsilon}\mathbf{P}_{-\epsilon}^{\mathrm{cnt}}$ as global scaling does not alter the magnitude spectrum of a matrix. In mathematical terms, the beamspace matrix corresponding to $\mathbf{G} \odot \mathbf{Q}_{\epsilon}^{\mathrm{cnt}}$ is also a $(-\epsilon,-\epsilon)$ shifted version of the true beamspace. 
\par Similar to the p-trajectory that is consistent with $\mathbf{P}_{\epsilon}^{\mathrm{cnt}}$, a trajectory that results in a beamspace shift of $(-\epsilon, -\epsilon)$ is designed using $\mathbf{Q}_{\epsilon}^{\mathrm{cnt}}$. A subsampling trajectory of length $2N-1$ that is consistent with $\mathbf{Q}_{\epsilon}^{\mathrm{cnt}}$ is one that sequentially traverses from the $\left( 2N-1 \right) ^{\mathrm{th}}$ contour to the $0^{\mathrm{th}}$ contour. The consistency of the trajectory can be observed by writing down the entries of the matrix $\mathbf{Q}_{\epsilon}^{\mathrm{cnt}}$. The trajectory is constrained to sample exactly one coordinate from any contour, and is called as the n-trajectory as it induces negative beamspace shift. A possible candidate for the n-trajectory is one that traverses exactly along the same path as the selected p-trajectory, but in the opposite direction. This choice, however, does not achieve diversity in sampling the virtual channel matrix, i.e., the p- and n-trajectories result in the same set of samples of $\mathbf{G}$ after correcting the phase errors due to CFO. Therefore, Swift-Link's p- and n-trajectories are chosen independently. 
\begin{figure}[htbp] 
\vspace{-3mm}
\subfloat[Swift-Link's p-trajectory]{\includegraphics[width=4.25cm, height=4.25cm]{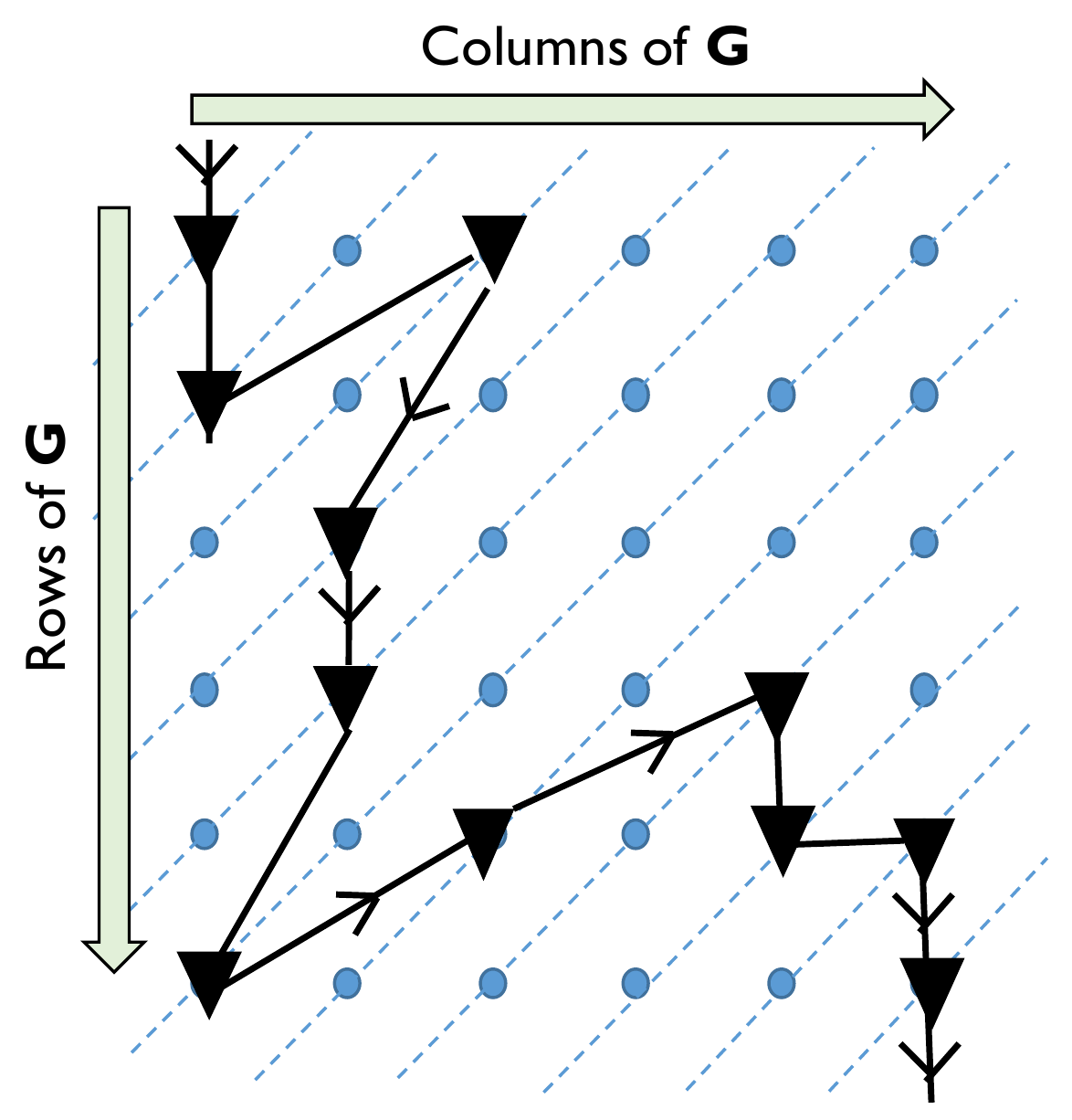}\label{fig:pseq}}
\hspace*{\fill}
\subfloat[Swift-Link's type I trajectory]{\includegraphics[width=4.25cm, height=4.25cm]{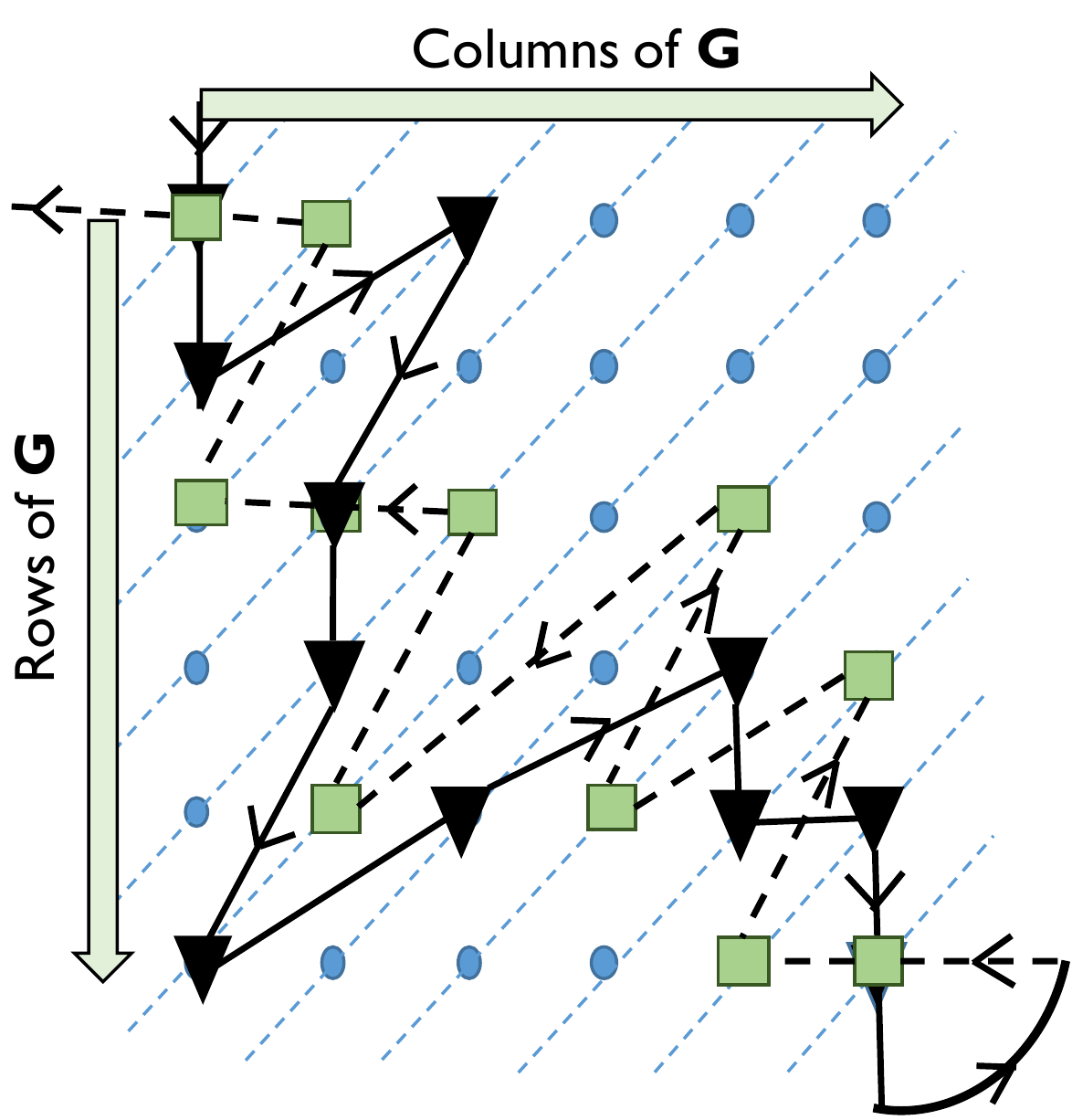}}\label{fig:typeI}
\caption{ \small A realization of Swift-Link's p-trajectory and type I trajectory. The type I trajectory comprises of a p-trajectory and a n-trajectory.\normalsize}
\vspace{-5mm}
\end{figure}
\subsection{Swift-Link's type I trajectory and shift correction}
We use $\Omega_p$ and $\Omega_n$ to denote the ordered set of coordinates sampled by the p- and n- trajectories. The type I trajectory $\Omega_{\mathrm{I}}$, first traverses according to $\Omega_p$ and then $\Omega_n$, i.e., the ordered set of coordinates sampled by type I trajectory is $\Omega_{\mathrm{I}}=\left\{\Omega_{p}, \Omega_{n} \right\} $. For a CFO of $\epsilon$, the channel measurements corresponding to the p- and n-components of $\Omega_{\mathrm{I}}$ are noisy versions of $\mathcal{P}_{\Omega_{p}}\left(\mathbf{G}\odot\mathbf{P}_{\epsilon}^{\mathrm{\mathrm{cnt}}}\right)$ and $e^{\mathrm{j} \theta_0}\mathcal{P}_{\Omega_{n}}\left(\mathbf{G}\odot\mathbf{P}_{-\epsilon}^{\mathrm{\mathrm{cnt}}}\right)$, where $\theta_0$ accounts for the intial phase of the n-trajectory. Specifically, $\theta_0 = (2N-1) \epsilon  +2(N-1)\epsilon$ for a $2N-1$ length p-trajectory. The first term in $\theta_0$ represents the phase error accumulated after traversing along the $2N-1$ length p-trajectory. The second term in $\theta_0$, i.e., $2(N-1)\epsilon$, accounts for the global phase difference between $\mathbf{P}_{-\epsilon}^{\mathrm{\mathrm{cnt}}}$ and $\mathbf{Q}_{\epsilon}^{\mathrm{\mathrm{cnt}}}$. As the matrices $\mathbf{G}\odot\mathbf{P}_{\epsilon}^{\mathrm{\mathrm{cnt}}}$ and $e^{\mathrm{j} \theta_0}\mathbf{G}\odot\mathbf{P}_{-\epsilon}^{\mathrm{\mathrm{cnt}}}$ are sparse in the 2D Fourier representation, CS can be used to recover them with undersampled measurements, and let $\mathbf{G}_{\mathrm{I,p}}$ and $\mathbf{G}_{\mathrm{I,n}}$ be the CS estimates. For $\mathbf{G}=\mathbf{1}\mathbf{1}^T$ and $M=2(2N-1)$ channel measurements acquired using binomial sampling-based type I trajectory, the masked beamspace estimates obtained using a standard CS algorithm under a CFO error are shown in Fig. \ref{fig:p_n_master}. The masked beamspace estimates in Fig. \ref{fig:p_n_master} are 2D Fourier transforms of the virtual channel estimates $\mathbf{G}_{\mathrm{I,p}}$ and $\mathbf{G}_{\mathrm{I,n}}$.  
\begin{figure}[htbp]
\subfloat[Masked beamspace estimate with CS over the p-trajectory samples]{\includegraphics[trim=1cm 0cm 1cm 1cm,clip=true,width=4.25cm, height=4.25cm]{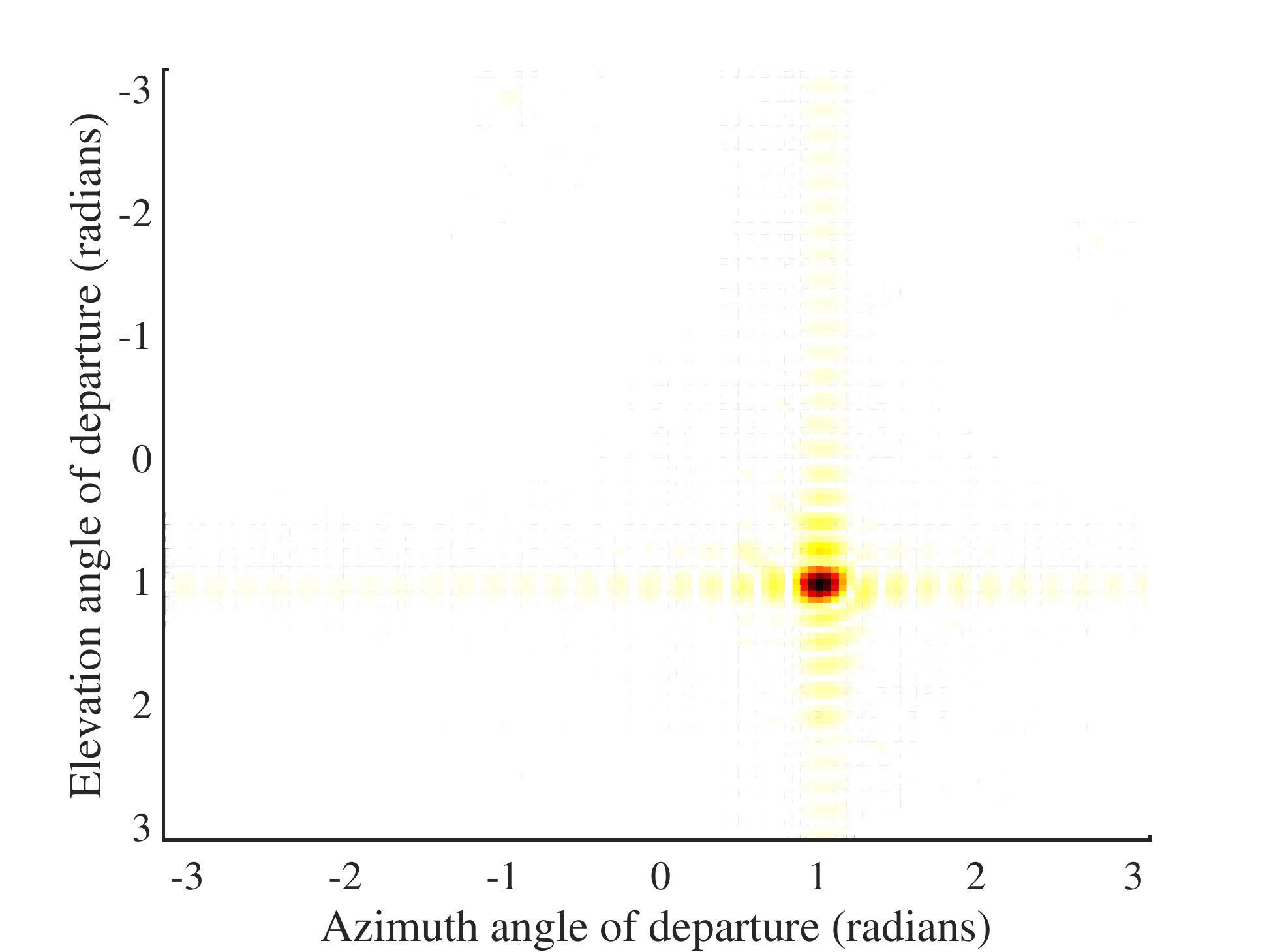}\label{fig:ptraj_recov}} 
\hspace*{\fill}
\subfloat[Masked beamspace estimate with CS over the n-trajectory samples]{\includegraphics[trim=1cm 0cm 1cm 1cm,clip=true,width=4.25cm, height=4.25cm]{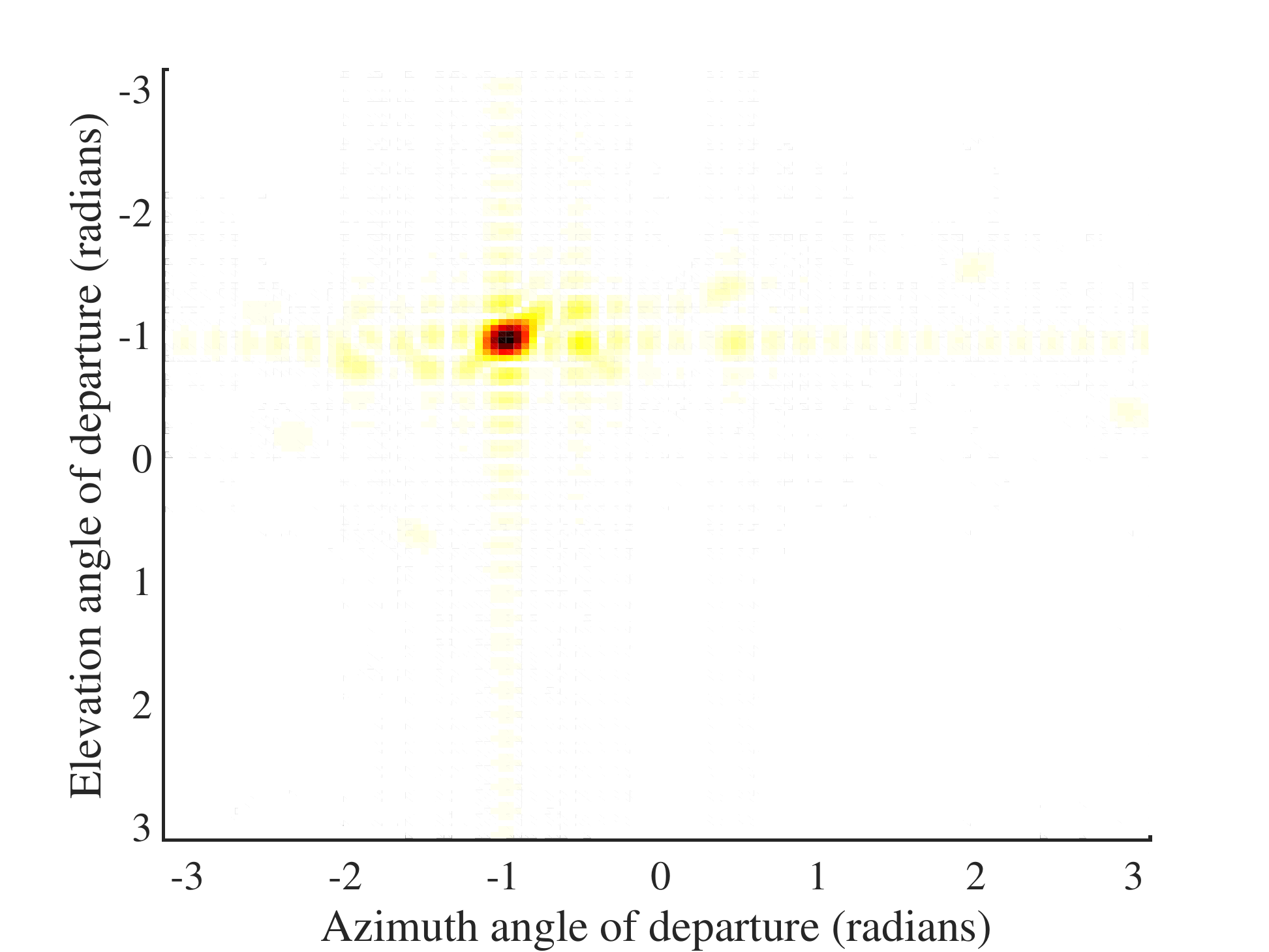}\label{fig:ntraj_recov}}
\caption{\small A CFO of $\epsilon$ induces a beamspace shift of $(\epsilon,\epsilon)$  for standard CS with the p-trajectory. Similarly, a beamspace shift of $(-\epsilon, -\epsilon)$ is induced with the n-trajectory. In this example, the true masked beamspace had a unique non-zero component at $(0,0)$ and $\epsilon=1 \, \mathrm{rad}$.
\normalsize}
\label{fig:p_n_master}
\end{figure}
The spatial shift between the masked beamspace estimates can be observed to be $(2\epsilon, 2\epsilon)$ for a CFO of $\epsilon$.
\par When CS perfectly recovers the matrices $\mathbf{G}\odot\mathbf{P}_{\epsilon}^{\mathrm{\mathrm{cnt}}}$ and $e^{\mathrm{j} \theta_0}\mathbf{G}\odot\mathbf{P}_{-\epsilon}^{\mathrm{\mathrm{cnt}}}$ under a CFO error, we have $\mathbf{G}_{\mathrm{I,p}}=e^{\mathrm{j} \theta }\mathbf{G}_{\mathrm{I,n}} \odot \mathbf{P}_{2 \epsilon}^{\mathrm{\mathrm{cnt}}}$, where $\theta= -\theta_0$. For the noisy case, $\epsilon$ is estimated from $\mathbf{G}_{\mathrm{I,p}}$ and $\mathbf{G}_{\mathrm{I,n}}$ using 
\begin{align}
\nonumber
(\hat{\epsilon}, \hat{\theta}) &=\underset{\Delta, \theta}{\mathrm{arg\,min}}\Vert \mathbf{G}_{\mathrm{I,p}}-e^{\mathrm{j} \theta}\mathbf{G}_{\mathrm{I,n}}\odot\mathbf{P}_{2\Delta}^{\mathrm{\mathrm{cnt}}}\Vert _{F}^{2}\\
\nonumber
&=\underset{\Delta,\theta}{\mathrm{arg\,min}} \big\{ \Vert\mathbf{G}_{\mathrm{I,p}}\Vert_{F}^{2}+\Vert\mathbf{G}_{\mathrm{I,n}}\Vert_{F}^{2}\\
&-2\mathrm{Re}\{\langle\mathbf{G}_{\mathrm{I,p}},e^{\mathrm{j}\theta}\mathbf{G}_{\mathrm{I,n}}\odot\mathbf{P}_{2\Delta}^{\mathrm{\mathrm{cnt}}}\rangle\} \big\} . 
\label{eq:cfestim_pre1}
\end{align}
As the first two terms in the minimization in \eqref{eq:cfestim_pre1} are independent of both $\Delta$ and $\theta$, the solution to \eqref{eq:cfestim_pre1} can be found using  
\begin{align}
(\hat{\epsilon}, \hat{\theta})&=\underset{\Delta,\theta}{\mathrm{arg\,max}}\,\,\mathrm{Re}\left\{ \left\langle \mathbf{G}_{\mathrm{I,p}},e^{\mathrm{j} \theta}\mathbf{G}_{\mathrm{I,n}}\odot\mathbf{P}_{2\Delta}^{\mathrm{\mathrm{cnt}}}\right\rangle \right\} \\
&=\underset{\Delta,\theta}{\mathrm{arg\,max}}\,\,\mathrm{Re}\left\{ e^{-\mathrm{j} \theta}\left\langle \mathbf{G}_{\mathrm{I,p}}\odot\mathbf{\overline{G}}_{\mathrm{I,n}},\mathbf{P}_{2\Delta}^{\mathrm{\mathrm{cnt}}}\right\rangle \right\}.
\label{eq:opt_cfo_init}
\end{align}
The maximization in \eqref{eq:opt_cfo_init} can be decoupled to obtain  
\begin{equation}
\label{eq:2Dsearch}
\hat{\epsilon}=\underset{\Delta}{\mathrm{arg\,max}}| \langle\mathbf{G}_{\mathrm{I,p}}\odot\mathbf{\overline{G}}_{\mathrm{I,n}},\mathbf{P}_{2\Delta}^{\mathrm{\mathrm{cnt}}}\rangle |,
\end{equation} 
which is a problem of finding the location of the maximum 2D frequency component of $\mathbf{G}_{\mathrm{I,p}}\odot\mathbf{\overline{G}}_{\mathrm{I,n}}$ on the $y=x$ line. With the definition in \eqref{eq:crossmat}, the inner product in \eqref{eq:2Dsearch} can be expressed as
\begin{align}
\langle\mathbf{G}_{\mathrm{I,p}}\odot\mathbf{\overline{G}}_{\mathrm{I,n}},\mathbf{P}_{2\Delta}^{\mathrm{\mathrm{cnt}}}\rangle&=\!\!\!\!\! \!\!\!\! \sum_{(r,c)\in\mathcal{I}_{N}\times\mathcal{I}_{N}}\!\!\!\!\!\!\!\!\mathbf{G}_{\mathrm{I,p}}(r,c)\mathbf{\overline{G}}_{\mathrm{I,n}}(r,c)e^{-\mathsf{j}2\Delta(r+c)} \\
\label{eq:expandedinnerprod}
&=\underset{k\in\mathcal{I}_{2N-1}}{\sum}\!\!\!\!\!e^{-\mathsf{j}2\Delta k} \!\!\!\! \!\!\!\! \underset{(r,c):r+c=k}{\sum}\!\!\!\!\!\!\mathbf{G}_{\mathrm{I,p}}(r,c)\mathbf{\overline{G}}_{\mathrm{I,n}}(r,c).
\end{align}
We use the inner product expansion in \eqref{eq:expandedinnerprod} to transform the 2D frequency search problem in \eqref{eq:2Dsearch} into a 1D search problem. To achieve such reduction in dimensionality, we define a vector $\mathbf{g} \in \mathbb{C}^{2N-1}$ such that its $k^{\mathrm{th}}$ element is
\begin{equation}
g\left[k \right]=\underset{\left(r,c\right):r+c=k}{\sum}\mathbf{G}_{\mathrm{I,p}}\left(r,c\right)\mathbf{\overline{G}}_{\mathrm{I,n}}\left(r,c\right). 
\label{eq:cfo_estim_vect}
\end{equation}
From  \eqref{eq:2Dsearch}, \eqref{eq:expandedinnerprod} and \eqref{eq:cfo_estim_vect}, it can be observed that 
\begin{equation}
\label{eq:cfestim_1D}
\hat{\epsilon}=\underset{\Delta}{\mathrm{arg\,max}}\,\, \left| \sum_{k=0}^{2N-2} \!\! g[k] e^{-\mathsf{j}2 \Delta k} \right|.
\end{equation}
The estimate $\hat{\epsilon}$ is obtained using an oversampled DFT of $\mathbf{g}$ followed by an interpolation technique proposed in \cite{candan}.
\par Using the estimated CFO, the masked beamspace can be estimated by correcting the CFO induced shift in $\mathbf{G}_{\mathrm{I,p}}$ and $\mathbf{G}_{\mathrm{I,n}}$. We define the compensated virtual channel estimates as $\mathbf{M}_{\mathrm{I,p}}=\mathbf{G}_{\mathrm{I,p}} \odot  \mathbf{P}_{-\hat{\epsilon}}^{\mathrm{\mathrm{cnt}}}$ and $\mathbf{M}_{\mathrm{I,n}}=\mathbf{G}_{\mathrm{I,n}} \odot  \mathbf{P}_{\hat{\epsilon}}^{\mathrm{\mathrm{cnt}}}$. The matrices $\mathbf{M}_{\mathrm{I,p}}$ and $\mathbf{M}_{\mathrm{I,n}}$ are noisy versions of the same underlying virtual channel matrix, upto a global phase. The global phase error between the compensated virtual channel estimates can be computed as $\phi=\underset{\omega}{\mathrm{arg\,min}}\,\Vert \mathbf{M}_{\mathrm{I,p}}-e^{\mathrm{j} \omega}\mathbf{M}_{\mathrm{I,n}}\Vert_F$. The scalar $\phi$ in closed form is $\mathrm{phase}(\langle \mathbf{M}_{\mathrm{I,p}}, \mathbf{M}_{\mathrm{I,n}} \rangle)$. The virtual channel estimate upto a scale factor is then $\mathbf{M}_{\mathrm{I}}= \mathbf{M}_{\mathrm{I,p}}+e^{\mathrm{j} \phi}\mathbf{M}_{\mathrm{I,n}}$. The matrix $\mathbf{M}_{\mathrm{I}}$ is then transformed to the antenna domain using the spectral mask concept and used for beamforming. For type I trajectory-based channel acquisition, we summarize Swift-Link in Algorithm \ref{alg:swlink}. We also provide a graphical illustration of Swift-Link in Fig. \ref{fig:theswiftlinkway}. 
\begin{figure*}[h]
\centering
\includegraphics[width=0.9 \textwidth]{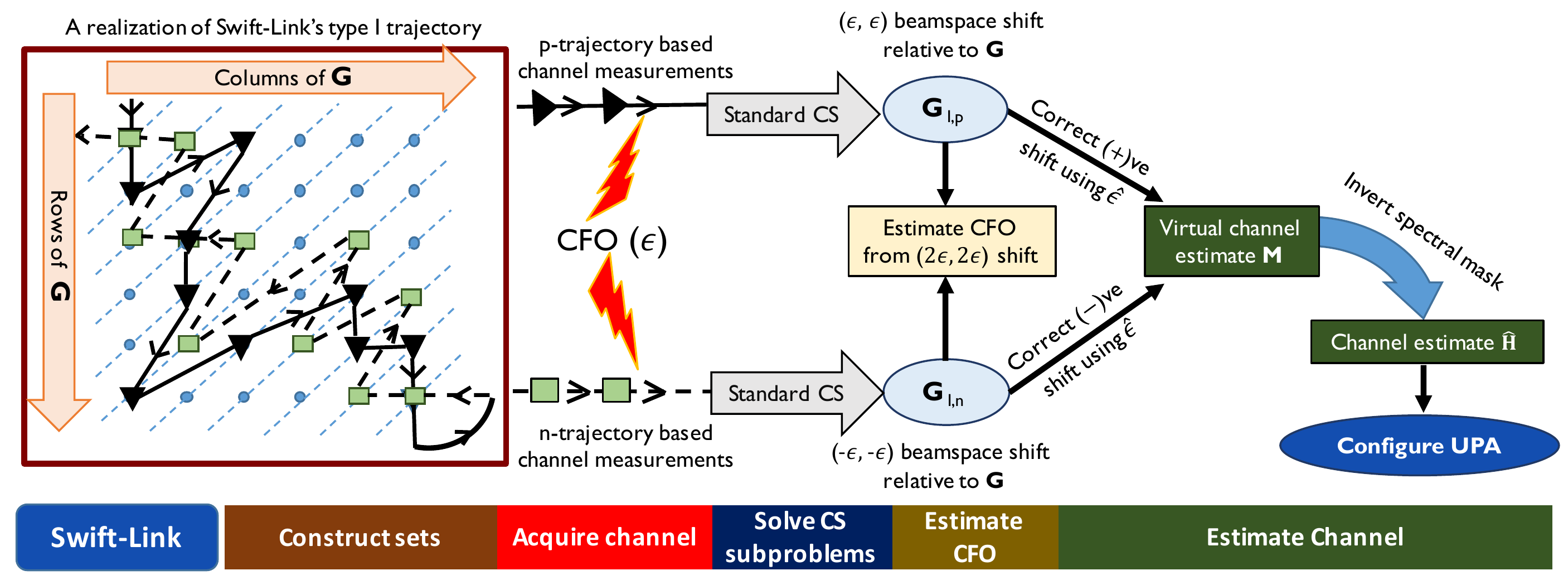}
  \caption{\small An illustration of Swift-Link for type I trajectory-based channel acquisition. Standard CS is performed using the channel measurements acquired by the p- and n- components of the type I trajectory. The CS estimates are then coherently combined to obtain $\hat{\mathbf{H}}$. \normalsize}
  \label{fig:theswiftlinkway}
\end{figure*}
\begin{algorithm}
\algrenewcommand\algorithmicindent{0.5em}%
\caption{Swift-Link for type I channel acquisition}\label{alg:swlink}
\begin{flushleft}
        \textbf{Inputs:} $M$, $N$ and a Zadoff-Chu sequence $\mathbf{z} \in \mathbb{C}^N$.\\
        \textbf{Outputs:} CFO estimate $\hat{\epsilon}$ and channel estimate $\hat{\mathbf{H}}$.
\end{flushleft}
\begin{algorithmic}
\Procedure{ConstructSets}{$M,N$}
\State Generate a p-trajectory of length $M/2$ on an $N\! \times \! N$ grid.
\State Similarly, generate an n-trajectory of length $M/2$.
\State Use the coordinates in the p- and n- trajectories to construct 
\State $\Omega_{p}=\{(r[n],c[n])\}^{M/2}_{n=1}$ and $\Omega_{n}=\{(r[n],c[n])\}^{M}_{n=M/2+1}$.
\State \textbf{return} $\Omega_{p},\Omega_{n}$
\EndProcedure
\Procedure{AcquireChannel}{$\Omega_{p},\Omega_{n},\mathbf{z}$}
\State Construct an $N\times N$ matrix $\mathbf{Z}=\left[\mathbf{z},\mathbf{Jz},\mathbf{J}^{2}\mathbf{z},...,\mathbf{J}^{N-1}\mathbf{z}\right]$.
\State For $n \in \mathcal{I}_{M}$, compute $\mathbf{b}_{n}=\mathbf{Z}\mathbf{e}_{r[n]}$ and $\mathbf{d}_{n}=\mathbf{Z}\mathbf{e}_{c[n]}$. 
\State Apply $\mathbf{b}_{n}\mathbf{d}^T_{n}$ to the UPA to get $y[n]$ under a CFO error.
\State \textbf{return} $\mathbf{y}$
\EndProcedure
\Procedure{SolveCSsubproblems}{$\mathbf{y},\Omega_{p},\Omega_{n}$}
\State Find $\mathbf{G}_{\mathrm{I},p}$ using 2D-DFT CS with $\Omega_{p}$ and $\{y[n]\}^{M/2}_{n=1}$.
\State Similarly, estimate $\mathbf{G}_{\mathrm{I},n}$ using $\Omega_{n}$ and $\{y[n]\}^{M}_{n=M/2+1}$.
\State \textbf{return} $\mathbf{G}_{\mathrm{I},p},\mathbf{G}_{\mathrm{I},n}$
\EndProcedure
\Procedure{EstimateCFO}{$\mathbf{G}_{\mathrm{I},p},\mathbf{G}_{\mathrm{I},n}$}
\State Construct $\mathbf{g} \in \mathbb{C}^{2N-1}$ from $\mathbf{G}_{\mathrm{I},n}$ and $\mathbf{G}_{\mathrm{I},p}$ using \eqref{eq:cfo_estim_vect}. 
\State Compute CFO estimate $\hat{\epsilon}$ using \eqref{eq:cfestim_1D}.
\State \textbf{return} $\hat{\epsilon}$
\EndProcedure
\Procedure{EstimateChannel}{$\mathbf{G}_{\mathrm{I},p},\mathbf{G}_{\mathrm{I},n},\hat{\epsilon}$}
\State Compute $\mathbf{M}_{\mathrm{I},p}=\mathbf{G}_{\mathrm{I},p}\odot\mathbf{P}_{-\hat{\epsilon}}^{\mathrm{cnt}}$ and $\mathbf{M}_{\mathrm{I},n}=\mathbf{G}_{\mathrm{I},n}\odot\mathbf{P}_{\hat{\epsilon}}^{\mathrm{cnt}}$.
\State Define $\phi=\mathrm{phase}(\langle\mathbf{M}_{\mathrm{I},p},\mathbf{M}_{\mathrm{I},n}\rangle)$, $\mathbf{M}=\mathbf{M}_{\mathrm{I},p}+e^{\mathsf{j}\phi}\mathbf{M}_{\mathrm{I},n}$.
\State Invert spectral mask to get $\hat{\mathbf{H}}=\mathbf{U}_{N}^{\ast}\mathbf{\boldsymbol{\Lambda}}_{\mathbf{z}}\mathbf{U}_{N}\mathbf{M}\mathbf{U}_{N}\mathbf{\boldsymbol{\Lambda}}_{\mathbf{z}}\mathbf{U}_{N}^{\ast}$.
\State \textbf{return} $\hat{\mathbf{H}}$
\EndProcedure
\end{algorithmic}
\end{algorithm}
\par A seemingly better strategy compared to coherent combining, i.e., $\mathbf{M}=\mathbf{M}_{\mathrm{I},p}+e^{\mathsf{j}\phi}\mathbf{M}_{\mathrm{I},n}$, in Algorithm \ref{alg:swlink} is to correct CFO in all the channel measurements obtained with Swift-Link's $\Omega_{\mathrm{I}}$ trajectory using $\hat{\epsilon}$.  The compensated channel measurements, i.e., $\{y[n]e^{-\mathsf{j}\hat{\epsilon}n} \}_{n \in \mathcal{I}_M}$, can be used for compressive channel estimation. CS-based algorithms that use these samples, however, do not perform well as any small CFO estimation error can result in large relative phase errors between the p- and n- components of $\Omega_{\mathrm{I}}$. For instance, the phase error between the two components would be $M\left(\epsilon-\hat{\epsilon}\right)/2$, where $M$ is the length of the type I trajectory. Coherent combining of the samples from the two components may require sparse self-calibration techniques \cite{biconvex} and can result in high complexity.  
\subsection{Swift-Link's type II trajectory and shift correction}
In this section, we propose a type II trajectory that does not introduce significant phase errors between the p- and n- components. The type II trajectory achieves robustness to CFO errors by interleaving the p- and n- trajectories in sequence. Specifically, the sequence of coordinates traversed is 
\normalsize
\begin{align}
\nonumber
\Omega_{\mathrm{II}}=
\{ \Omega_{p}(1),\Omega_{n}(1),\Omega_{p}(2),\Omega_{n}(2),..,\Omega_{p}(M/2),\Omega_{n}(M/2)\},
\end{align}
\normalsize where $\Omega \left( i \right)$ denotes the $i^{\mathrm{th}}$ coordinate in the ordered set $\Omega$. Due to interleaving, CFO induces a phase error of $2\epsilon$ in every step along each of the p- and n- components of $\Omega_{\mathrm{II}}$. 
\par Swift-Link recovers $\mathbf{G} \odot \mathbf{P}^{\mathrm{cnt}}_{2 \epsilon}$ and $e^{\mathrm{j} \beta}\mathbf{G} \odot \mathbf{P}^{\mathrm{cnt}}_{-2 \epsilon}$ using CS over the odd and even samples of the $\Omega_{\mathrm{II}}$ trajectory.    
Let $\mathbf{G}_{\mathrm{II,p}}$ and $\mathbf{G}_{\mathrm{II,n}}$ be the virtual channel estimates corresponding to the odd and even components. Similar to the type I case, the CFO is estimated as   
\begin{equation}
\hat{\epsilon}=\underset{\Delta}{\mathrm{arg\,max}}\left|\left\langle \mathbf{G}_{\mathrm{II,p}}\odot\mathbf{\overline{G}}_{\mathrm{II,n}},\mathbf{P}_{4\Delta}^{\mathrm{\mathrm{cnt}}}\right\rangle \right|.
\end{equation}
The estimated CFO $\hat{\epsilon}$ is used to correct phase errors in all the channel measurements obtained using $\Omega_{\mathrm{II}}$ and CS-based channel estimation is performed using the compensated samples. Although interleaving can efficiently leverage diversity in the channel measurements, it reduces the CFO correction range by a factor of $2$ compared to the type I case. 
\par In this section, we have illustrated our algorithm for type I and type II trajectories of length $M=4N-2$. Nevertheless, the same design principles can be applied for $M< 4N-2$ by selecting the trajectory about the $( N-1 ) ^{\mathrm{th}}$ contour. In this case, the length of each of the p- and n-trajectories is set to $M/2$. For $M<4N-2$, the p-trajectory would begin at the $( N-1 -\lfloor M/4 \rfloor )^{\mathrm{th}}$ contour and end at the $( N-2 + \lceil M/4 \rceil  )^{\mathrm{th}}$ one. The n-trajectory would traverse through the same contours, but in the opposite direction. The trajectories are selected about the $\left(N-1 \right) ^{\mathrm{th}}$ contour to ensure maximum randomization. The phase errors $\{e^{\mathsf{j}\epsilon k}\}_{k=0}^{M/2-1}$, when filled along the $M/2$ length p-trajectory, can be observed to be consistent with $e^{-\mathsf{j}\theta_{p,M}}\mathbf{P}_{\epsilon}^{\mathrm{cnt}}$, where $\theta_{p,M}=N-1 - \lfloor M/4 \rfloor$. Channel acquisition using such trajectory results in an undersampled version of $e^{-\mathsf{j}\theta_{p,M}} \mathbf{G}\odot \mathbf{P}_{\epsilon}^{\mathrm{cnt}}$, i.e., $e^{-\mathsf{j}\theta_{p,M}} \mathcal{P}_{\Omega_p}( \mathbf{G}\odot \mathbf{P}_{\epsilon}^{\mathrm{cnt}})$. Similarly, the n-trajectory acquires $e^{-\mathsf{j}\theta_{n,M}} \mathcal{P}_{\Omega_n}( \mathbf{G}\odot \mathbf{P}_{-\epsilon}^{\mathrm{cnt}})$ for some $\theta_{n,M}$. It can be noticed that the sparse beamspace matrices corresponding to the samples acquired by the p- and n-trajectories have the same physical interpretation as before, i.e., they are $(\epsilon, \epsilon)$ and $(-\epsilon, -\epsilon)$ shifted versions of the actual masked beamspace. Therefore, Swift-Link can still be applied with the designed training for $M<4N-2$ channel measurements. Extension of our algorithm for $M>4N-2$ requires the design of trajectories that are longer than the proposed trajectories and achieve robustness to CFO. Designing such trajectories for low-complexity beam alignment is an interesting direction for future work.
\par The fundamental principle of Swift-Link is to design a robust CS measurement matrix in a way that aids low complexity algorithms. Swift-Link's algorithm for the type II case requires solving two CS problems in dimension $N^2$ with $M/2$ measurements and one CS problem in dimension $N^2$ with $M$ measurements. In contrast, the tensor or lifting-based methods in \cite{nitinanalog} and \cite{biconvex} require solving an $N^2M$ dimensional CS problem with $M$ measurements. Due to the use of large antenna arrays at mmWave, the number of variables in  tensor-based optimization \cite{nitinanalog} can be in the order of millions for typical mmWave settings. As a result, the techniques proposed in \cite{nitinanalog} and \cite{biconvex} require a high computational complexity when compared to Swift-Link. For a greedy CS algorithm like OMP \cite{GOMP}, the computational complexity of Swift-Link and the tensor-based method is $\mathcal{O}(N^2 M)$ and $\mathcal{O}(N^2 M^2)$.        
\subsection{Analysis of Swift-Link} \label{sec:swlink_ana}
Swift-Link's success depends on the recovery of the shifted beamspace matrices using each of the p- and n-trajectories. In this section, we consider the samples acquired with the p-trajectory and derive novel reconstruction guarantees for the masked beamspace recovery under the CFO perturbation.
\par We describe the RIP and provide definitions required for our analysis. Consider a $K$ sparse masked beamspace matrix $\mathbf{S}$ with $\left(x_{k},y_{k}\right)_{k\in \mathcal{I}_K}$ as the locations of non-zero entries. For each $k \in \mathcal{I}_K$, $\left(x_{k},y_{k}\right) \in \mathcal{I}_N \times \mathcal{I}_N$ and let $\mathbf{S}\left( x_k, y_k \right)=s_k$. The virtual channel matrix coordinates can be expressed as
\begin{equation}
\mathbf{G}\left(r,c\right)=\sum_{k\in\mathcal{I}_K}s_{k}e^{\mathrm{j}2\pi \left(rx_{k}+cy_{k}\right)/N}.
\label{eq:analysis_G}
\end{equation}   
Let $M_{p}=\left|\Omega_{p}\right|$ be the number of samples of $\mathbf{G}$ acquired by the p-trajectory under the CFO perturbation. The $M_p$ length p-trajectory must begin at $p_0=N-(M_p+1)/2$ contour as it is chosen about the $\left(N-1\right)^{\mathrm{th}}$ contour. We assume that $M_p$ is an odd integer in our analysis. The CS matrix corresponding to $\Omega_p$  satisfies the RIP of order $K$ with parameter $\delta_K$ \cite{csintro} if the channel measurements corresponding to $\mathbf{S}$ satisfy 
\begin{equation}
\left(1-\delta_{K}\right)\left\Vert \mathbf{S}\right\Vert _{F}^{2}\leq\frac{1}{M_p}\left\Vert \mathbf{y}\right\Vert ^{2}\leq\left(1+\delta_{K}\right)\left\Vert \mathbf{S}\right\Vert _{F}^{2}
\label{eq:RIPdefn}
\end{equation}
for every $K$ sparse matrix $\mathbf{S}$. The normalization by $1/M_p$ in \eqref{eq:RIPdefn} ensures that the columns of the CS matrix have unit norm.
A small RIP parameter implies that the distance between sparse matrices is preserved in the measurement space and allows sparse recovery \cite{csintro}.
\par Now, we derive a sufficient condition that guarantees the recovery of an $\left(\epsilon, \epsilon\right)$ shifted version of the masked beamspace matrix using Swift-Link's p-trajectory. We define the minimum pair-wise spacing as $d_{\mathrm{min}}=\underset{k,\ell\in \mathcal{I}_K,\,k\neq\ell}{\mathrm{min}}\mathrm{max}\left\{ \left|x_{k}-x_{\ell}\right|^{+},\left|y_{k}-y_{\ell}\right|^{+}\right\} $, where 
\begin{equation}
\left|x\right|^{+}=\begin{cases}
\begin{array}{c}
\left|x\right|\\
N-\left|x\right|
\end{array} & \begin{array}{c}
\mathrm{if}\,\left|x\right|\leq N/2\\
\mathrm{otherwise}
\end{array}\end{cases}.
\end{equation} 
Notice that $1 \leq d_{\mathrm{min}} \leq  N/2$ for any set of non-sparse locations in $\mathbf{S}$.
\begin{theorem} 
On an average, the CS matrix corresponding to Swift-Link's p-trajectory satisfies the RIP of order $K>1$ with parameter $\delta_K$ when
\begin{equation}
M_p \geq \frac{K-1}{\delta_K \mathrm{sin}\frac{\pi d_{\mathrm{min}}}{N}}\left[4+2\mathrm{log}\left(\frac{N}{N+1-(M_{p}-1)/2}\right)\right]
\label{eq:CSmeasth}
\end{equation}
channel measurements are acquired using a uniform distribution over contours.  
\label{theorem1}
\end{theorem}
\begin{proof}
See Appendix-A. 
\end{proof}
From the lower bound in \eqref{eq:CSmeasth}, it can be concluded that $\mathcal{O}(2KN \mathrm{log}N)$ samples suffice for Swift-Link's CS matrix to satisfy the RIP in the worst case, i.e., when $d_{\mathrm{min}}=1$. As Swift-Link limits randomization within the contours for robustness to CFO, it requires $\mathcal{O}(N)$ times the number of measurements required by IID random switching; this corresponds to a subsampling ratio of $\mathcal{O}( \mathrm{log} N /N)$. The bound in \eqref{eq:CSmeasth}, however, may be weak as simulation results show that Swift-Link outperforms a non-coherent algorithm that requires sub-linear number of channel measurements. As the p-trajectory can sample a maximum of $2N-1$ coordinates with a single RF chain, the worst case bound suggests that the use of $\mathcal{O}(K \mathrm{log}N)$ RF-chains at the TX guarantees CFO robust beam alignment using Swift-Link. In this work, we focus on the beamforming architecture with a single RF chain and show that our algorithm still performs well through  simulations.    
\section{Extending Swift-Link to wideband systems} \label{sec:wbextend}
In this section, we describe a frame structure that allows applying Swift-link for beam alignment in  wideband mmWave systems. For quasi-omnidirectional reception at the RX, let $\mathbf{H}[\ell] \in \mathbb{C}^{N \times N}$ be the $\ell ^{\mathrm{th}}$ tap of the $L$ tap wideband MIMO channel between the UPA at the TX and the RX. At mmWave frequencies, each of these taps is approximately sparse in the 2D-DFT representation \cite{heathoverview}. Using the notation in Sec.~\ref{sec:sysmodel}, the received samples can be expressed as 
\begin{equation}
y[n]=e^{\mathrm{j} \epsilon n }\sum_{\ell=0}^{L-1}\mathbf{b}_{n-\ell}^{T}\mathbf{H}\left[\ell\right]\mathbf{d}_{n-\ell}t[n-\ell]+v[n].
\label{eq:wideband}
\end{equation}
At any time instant, the channel measurement in \eqref{eq:wideband} is a linear combination of the previous $L$ transmitted symbols with weights determined by appropriate beam training vectors. The model in \eqref{eq:wideband} suggests that  the phased array configuration $\mathbf{b}_n\mathbf{d}^T_n$ can be changed every symbol duration. In practice, however, changing the configuration every symbol duration may not be possible as phase shifters have a finite settling time. Hence, a time domain sequence spanning several symbol durations is transmitted using a beam training configuration \cite{kiranchannel}. Using several such configurations, different spatial projections of the channel matrix can be obtained for beam alignment.
\begin{figure}[h]
\vspace{-2mm}
\centering
\includegraphics[width=0.48 \textwidth]{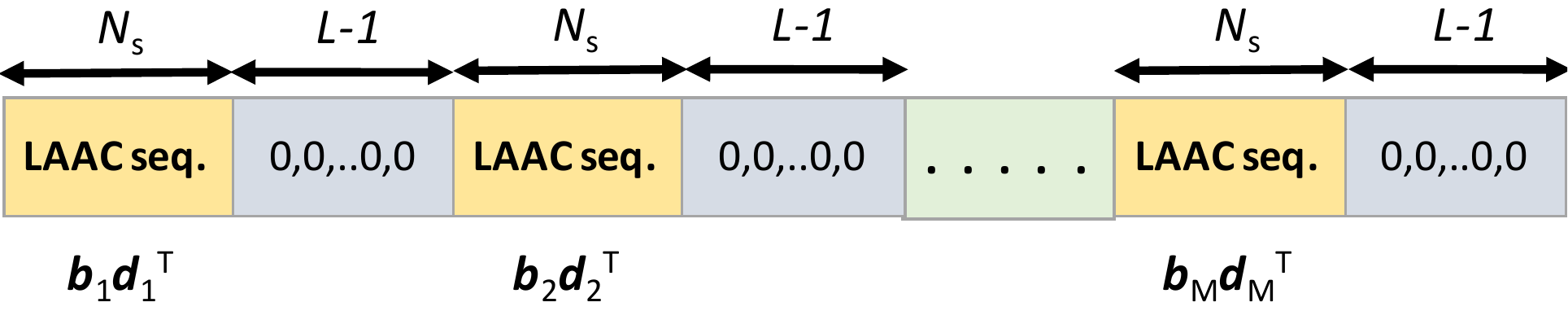}
  \caption{\small Frame structure used for beam training in a wideband system. A Low Aperiodic Auto-Correlation (LAAC) sequence is sent over different beam training configurations determined by Swift-Link. \normalsize}
  \label{fig:frame}
\end{figure}
The training sequences for different beam configurations are separated by a guard interval with zeros to avoid inter frame interference \cite{kiranchannel}. The settling time of the phase shifters is assumed to be smaller than the duration of the guard interval. The frame structure in Fig.~\ref{fig:frame}, uses Low Aperiodic Auto-Correlation (LAAC) sequences along the time dimension so that wideband channel estimation can be decoupled in the time dimension using the LAAC property. 
\par The TX transmits an $N_s$ length LAAC sequence and $L-1$ zeros across the channel for each of the phase shift configurations determined by Swift-Link's trajectory. Under the CFO impairment, the RX correlates the received signal with the same $N_s$ length sequence and measures the channel impulse response for each phase shift configuration at the TX. In this work, the sum of the responses obtained for each configuration is used for beam training, i.e., $\sum_{\ell \in \mathcal{I}_L} \mathbf{H}\left[\ell\right]$, the equivalent narrowband MIMO channel is considered for compressive beam alignment. As the channel measurements corresponding to consecutive beams are $N_s+L-1$ symbol durations apart, the effective CFO to be considered in beam training is $\left( N_s + L-1 \right)\epsilon$. The equivalent narrowband channel and CFO are estimated using Swift-Link. Then, the quantized versions of the singular vectors corresponding to the maximum singular value of the estimated channel matrix are used for beamforming at the TX \cite{nitinanalog}. It can be observed that beam alignment using the equivalent narrowband channel, i.e., $\sum_{\ell \in \mathcal{I}_L} \mathbf{H}\left[\ell\right]$ may result in sub-optimal beam alignment when the number of spatial paths within the delay spread are beyond Swift-Link's recoverable limit. In such scenarios, Swift-Link's framework can still be applied to recover each of the channel taps separately at the expense of higher complexity. 
\section{Swift-Link in action} \label{sec:swsim}
In this section, we evaluate Swift-Link for a practical wideband mmWave setting. We consider a hardware architecture in Fig.~\ref{fig:architect}, where the TX is equipped with a UPA of size $32 \times 32$ and the RX uses a fixed quasi-omnidirectional pattern for beam training. The phase shifters at the TX are assumed to have a resolution of $3$ bits. Our algorithm uses a core ZC sequence of length $32$ and root $11$ so that its $3$ bit phase quantized version satisfies the uniform DFT property in the magnitude sense. A disciplined approach to generate such unimodular ZAC sequences with quantized set of phases can be found in \cite{phasequant}. 
 We consider a mmWave carrier frequency of $28 \, \mathrm{GHz}$ and an operating bandwidth of $W=100 \, \mathrm{MHz}$ \cite{bandwidthref}, which corresponds to a symbol duration of $T=10\, \mathrm{ns}$.
\par The mmWave channel in our simulations was derived from the NYU channel simulator \cite{NYUSIM} for a TX-RX separation of $60 \, \mathrm{m}$ in an Urban-Micro Non-Line-of-Sight environment. The results we report are the averages over $100$ such channel realizations. For this setting, the omnidirectional RMS delay spread was found to be less than $43 \, \mathrm{ns}$ in more than $85 \%$ of the channel realizations. Considering the leakage effects due to pulse shaping, the wideband channel is modeled using $L=13$ taps corresponding to a duration of $130 \, \mathrm{ns}$. The channel is then scaled so that $\mathbb{E}\left[\sum_{\ell\in\mathcal{I}_{L}}\left\Vert \mathbf{H}\left[\ell\right]\right\Vert _{\text{F}}^{2}\right]=N^2$, where the expectation is taken across several channel realizations. 
\par According to the frame structure in Fig.~\ref{fig:frame}, the TX periodically transmits bipolar Barker sequences \cite{Barker} of length $N_s=13$ with different beam configurations determined by the algorithm. We use a bipolar sequence for the LAAC sequence so that its auto-correlation property is not heavily distorted by CFO. Under a CFO error, the RX correlates the received stream with the same Barker sequence and sums the channel response for each configuration. To recover the sparse channel components associated with each of the p- and n-trajectories of Swift-Link, we use EM-BG-AMP algorithm \cite{EMGAMP} as it has good phase transition properties compared to the other greedy algorithms. For a benchmark, we compare our algorithm with EM-BG-AMP with IID random phase shifts and zero CFO. We illustrate the vulnerability of CS to phase errors by evaluating EM-BG-AMP when IID random phase shifts are used for beam training at the TX. Furthermore, we also show the performance of a non-coherent compressive beam alignment method called Agile-Link \cite{agile}. We set the bin parameters of Agile-Link to $B_{\mathrm{az}}=B_{\mathrm{el}}=4$, and the number of hashes as $M/ \left( B_{\mathrm{az}}B_{\mathrm{el}} \right)$.   
For a fair comparison, beam training based on the equivalent narrowband channel is considered for all the algorithms. 
\begin{figure}[h]
\centering
\includegraphics[trim=2cm 6.25cm 2cm 7.4cm,clip=true,width=0.4 \textwidth]{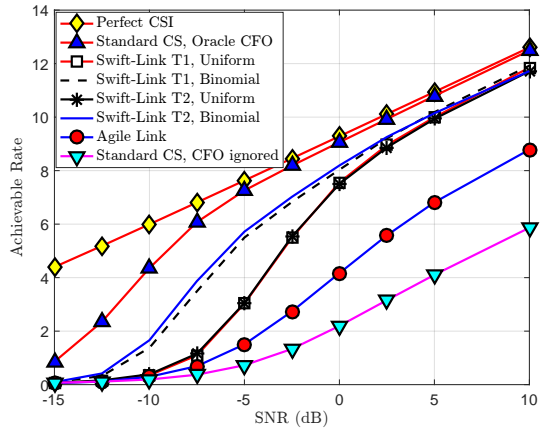}
  \caption{\small Achievable rate with SNR when $124$ beam training measurements are used by standard CS and Swift-Link, and $128$ by Agile-Link. Different values of CFO are used for the type I (T1) and type II (T2) methods to maintain the same induced beamspace shift.  \normalsize}
  \label{fig:SNR}
\end{figure}
\begin{figure}[h]
\centering
\includegraphics[trim=1cm 6.25cm 2cm 7.4cm,clip=true,width=0.4 \textwidth]{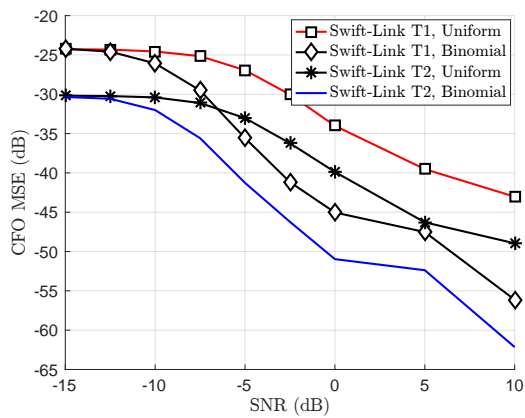}
  \caption{\small The mean squared error of the CFO estimate obtained from Swift-Link's beamspace shift estimation. The plot shows that sampling the coordinates according to a binomial distribution results in a better beamspace estimate when compared to the uniform case. \normalsize}
  \label{fig:CFOMSE}
\end{figure}
\par We use $\mathbf{f}_\mathrm{e}$ and $\mathbf{f}_\mathrm{a}$ to denote the beam alignment vectors estimated by an algorithm along the elevation and azimuth dimensions. The TX applies $\mathbf{f}_\mathrm{e}\mathbf{f}^T_\mathrm{a}$ to the UPA and the RX observes a wideband SISO channel with coefficients $\{\mathbf{f}_{\mathrm{e}}^{\ast}\mathbf{H}[\ell]\overline{\mathbf{f}}_{\mathrm{a}}\}_{\ell=0}^{L-1}$.  
The achievable rate is obtained using water filling based-power allocation \cite{tsewicomm} in the frequency domain for this equivalent wideband SISO channel seen after beam alignment. We define the SNR in the beam training measurements as $\mathrm{SNR}=10\mathrm{log}_{10}(1/\sigma^2)$. The algorithms are compared in terms of the achievable rate after beam alignment as a function of the received SNR, number of beam training measurements and CFO.
The impact of the residual CFO was ignored in computing the rate for all the algorithms. For $124$ beam training measurements, it can be noticed from Fig.~\ref{fig:SNR} that Swift-Link outperforms Agile-Link and EM-BG-AMP with IID random phase shifts at all levels of SNR. This is because Swift-Link exploits the structure in the unknown phase errors unlike other algorithms. For Fig.~\ref{fig:SNR}, Fig. ~\ref{fig:CFOMSE} and Fig.~\ref{fig:meas}, the analog domain CFO for the type I and type II trajectories was chosen differently as $800\, \mathrm{KHz}$ and $400\, \mathrm{KHz}$ to maintain the same beamspace shift effect in both cases. The offsets chosen are close to the maximum CFO correction limit of Swift-Link. The MSE of the CFO estimate is shown as a function of SNR in Fig.~\ref{fig:CFOMSE}. It can be observed from Fig.~\ref{fig:SNR} and Fig.~\ref{fig:CFOMSE} that sampling the coordinates on the contours according to the binomial distribution is better than the uniform case. 
\begin{figure}[h]
\centering
\includegraphics[trim=2cm 6.25cm 2cm 7.4cm,clip=true,width=0.4 \textwidth]{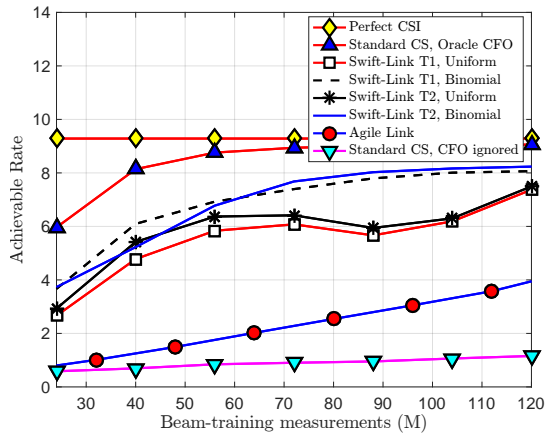}
  \caption{\small The plot shows the achievable rate with different number of beam training measurements for an SNR of $0$ dB. The decrease in the rate for uniform sampling is because the improvement in CFO estimate with channel measurements is not commensurate with the improvement in the channel estimate. \normalsize}
  \label{fig:meas}
\end{figure}
\par In Fig.~\ref{fig:meas}, we compute the achievable rate as a function of the number of beam training measurements for all the algorithms. It can be observed that Swift-Link performs better than the rest of them for any number of measurements. Swift-Link is limited in terms of beam training measurements as it can acquire a maximum of $2\left(2N-1\right)$ measurements over the $2N-1$ contours of an $N \times N$ matrix. For our setting, the number of measurements was sufficient enough to achieve a good rate. Developing new trajectories and low complexity correction strategies to recover channels in rich scattering environments is an interesting research direction. 
\begin{figure}[h]
\centering
\includegraphics[trim=2cm 6.25cm 2cm 7.4cm,clip=true,width=0.4 \textwidth]{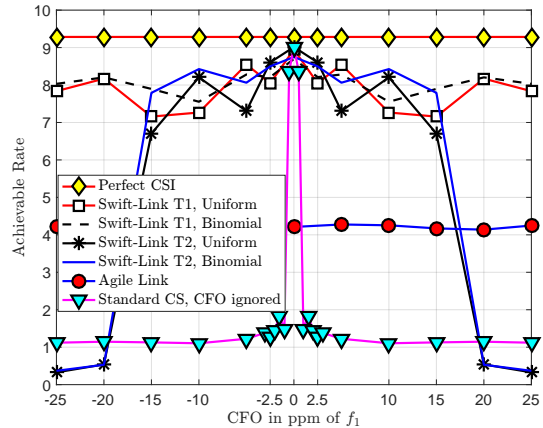}
  \caption{\small For an SNR of $0$ dB, the plot shows the achievable rate as a function of CFO when Swift-Link uses $124$ beam training measurements. Swift-Link performs well over a wide range of CFO. \normalsize}
  \label{fig:CFO}
\end{figure}
\par The CFO correction capability of Swift-Link depends on the type of trajectory used for beam alignment. For the type I trajectory, the induced beam-space shift can be uniquely determined if $\left|2\left(N_{s}+L-1\right)\epsilon\right|<\pi$. This requirement translates to a maximum correctable CFO of 
\begin{equation}
f_{\mathrm{max}}^{\mathrm{I}}=\frac{W}{4\left(N_{s}+L-1\right)},
\end{equation}
in the analog domain for the type I trajectory. As the induced beam-space shift is doubled with the type II trajectory due to interleaving, it supports a CFO correction limit of $f_{\mathrm{max}}^{\mathrm{II}}=f_{\mathrm{max}}^{\mathrm{I}}/2$. We compare the performance of all the algorithms for different values of CFO in Fig.~\ref{fig:CFO}. The number of beam training measurements were set to $124$ for Swift-Link, and $128$ for EM-BG-AMP with IID phase shifts and Agile-Link. It can be observed that standard CS algorithms fail even for a CFO of $1 \, \mathrm{ppm}$. As Agile-Link is a non-coherent method that considers just the magnitude of beam training measurements, its performance is invariant to CFO. It can be noticed that Swift-Link performs better than Agile-Link and standard CS within the supported range of CFO correction. The performance variation of Swift-Link within this range is because of the grid mismatch effects that arise due to CFO dependent beam-space shifts.
\par From a hardware design perspective, Swift-Link is advantageous over exhaustive beam search and Agile-Link as it requires a lower output dynamic range at the analog-to-digital converter. A higher dynamic range or peak-to-average-power-ratio (PAPR) of the received signal is not desirable as it can cause clipping effects after automatic gain control (AGC). As exhaustive beam scan uses highly directional beams, the variation in the received power across different beam scans is higher, especially for channels that are sparse in the beamspace. For example, when $\mathbf{H}=\mathbf{1}\mathbf{1}^T$, the maximum and minimum received power observed in a sequence of $N^2$ 2D-DFT based beam scans are $N^2$ and $0$. As Swift-Link uses ZC sequences in the antenna domain, the beams generated to acquire channel measurements are spread out in all directions. On the one hand, the maximum and minimum received power observed with Swift-Link's training are both $1$ for an all-ones $\mathbf{H}$. On the other hand, Agile-Link uses multi-armed beams that focus the transmit power in few directions and has a higher PAPR than Swift-Link. For the simulation setting with $M=124$ and $\mathrm{SNR}=0\,\mathrm{dB}$, we observed that the PAPR of the received signal for exhaustive beam search, Agile-Link and Swift-Link were $34.02 \, \mathrm{dB}$, $16.23\, \mathrm{dB}$ and $11.97 \,\mathrm{dB}$. We believe that Swift-Link marks an important step in translating compressed sensing technology to practical mmWave radios.
\section{Conclusions and future work} \label{sec:concl_fut}
In this work, we have studied the impact of CFO on compressed sensing based channel estimation and proposed a CFO robust beam-alignment algorithm called Swift-Link. Swift-Link uses circulantly shifted ZC sequences in the spatial domain and transforms the phased array architecture to a virtual switching architecture. We defined the concept of spectral mask that allows guaranteed channel reconstruction even with strict hardware constraints. In our future work, we will develop well conditioned spectral masks with ultra-low resolution phase shifters.
\par We define the concept of trajectory and determine the amount of circulant shift to be applied for each beam training measurement to minimize the impact of CFO on beam alignment. Swift-Link is advantageous over sparse self-calibration techniques in terms of memory and computational complexity, and has analytical guarantees. Our training sequence essentially aligns the sparsity basis of the calibration errors with the sparsity basis of the transformed signal. Furthermore, the impact of the calibration errors on CS-based recovery of the sparse signal was controlled by clever subsampling. Extending Swift-Link to generic sparse self-calibration problems is an interesting direction for future work. 
\section*{Appendix~A: Proof of Theorem \ref{theorem1}}
In this section, we derive a sufficient condition for Swift-Link's p-trajectory to satisfy the RIP on an average. For ease of notation, we label the $M_p$ channel measurements corresponding to the p-trajectory as $\left\{ y[n]\right\} _{n=p_{0}}^{p_{0}+M_{p}-1}$. As $r[n]+c[n]=n$ for every element on the $n^{\mathrm{th}}$ contour, the received samples can be expressed using  \eqref{eq:canonical} and \eqref{eq:analysis_G} as
\begin{equation}
y[n]=e^{-\mathrm{j} \epsilon p_{0}}\sum_{k\in \mathcal{I}_K}s_{k}e^{\mathrm{j} \left((\frac{2\pi x_k}{N}+\epsilon)r[n]+(\frac{2\pi y_k}{N}+\epsilon)c[n]\right)}.
\label{eq:samplesredif}
\end{equation}
Notice that we ignore noise in \eqref{eq:samplesredif} as the RIP constant is defined for a CS measurement matrix. The normalized energy of the sampled measurements is $E=\sum_{n=p_{0}}^{p_{0}+M_{p}-1}\left|y[n]\right|^{2}/M_p$. The result in Theorem \ref{theorem1} is derived using bounds on $\mathbb{E}_{\Omega_p}\left[E\right]$ for a sparse masked beamspace matrix $\mathbf{S}$. 
\par We define the pair-wise spacings between the non-zero locations of $\mathbf{S}$ as $\Delta x_{k,\ell}= x_{k}-x_{\ell}$ and $\Delta y_{k, \ell}= y_{k}-y_{\ell}$. For a given realization of the p-trajectory, the normalized energy of received samples can be expressed as 
\begin{align}
\nonumber
E&=\frac{1}{M_{p}}\sum_{n=p_{0}}^{p_{0}+M_{p}-1}\left|\sum_{k=1}^{K}s_{k}e^{\mathrm{j} (\epsilon+2\pi x_{k}/N)r[n]+j(\epsilon+2\pi y_{k}/N)c[n]}\right|^{2}\\
\nonumber
&=\frac{1}{M_{p}}\sum_{n=p_{0}}^{p_{0}+M_{p}-1}\sum_{k=1}^{K}\left|s_{k}\right|^{2}\\
&+\frac{1}{M_{p}}\sum_{n=p_{0}}^{p_{0}+M_{p}-1}\mathrm{Re}\Bigl\{ \underset{k\neq\ell}{\sum}s_{k}s_{\ell}^{\ast} 
e^{\mathrm{j} 2\pi\left(\Delta x_{k,\ell}r[n]+\Delta y_{k,\ell}c[n]\right)/N} \Bigr\}.
\label{eq:totalenergy_det}
\end{align}
Notice that $\sum_{k=1}^{K}\left|s_{k}\right|^{2}=\left\Vert \mathbf{S}\right\Vert _{F}^{2}$. When the trajectory is randomly chosen according to a uniform distribution, the p-trajectory acquires an energy of  
\begin{align}
\mathbb{E}_{\Omega_{p}}\left[E\right]&=\left\Vert \mathbf{S}\right\Vert _{F}^{2}+\frac{1}{M_{p}}\mathrm{Re} \Bigl\{ \underset{k\neq\ell}{\sum}s_{k}s_{\ell}^{\ast} \times \\
& \sum_{n=p_{0}}^{p_{0}+M_{p}-1}\mathbb{E}_{\Omega_{p}}[e^{\mathrm{j} 2\pi\left(\Delta x_{k,\ell}r[n]+\Delta y_{k,\ell}c[n]\right)/N}] \Bigr\}.
\label{eq:expected_energy}
\end{align}
on an average. To bound $\mathbb{E}_{\Omega_{p}}\left[E\right]$, we define the functions $T_{x,y}[n]=\mathbb{E}_{\Omega_{p}}\left[e^{\mathrm{j} \left(x r[n]+ y c[n]\right)}\right]$ and $B_{x,y}=\sum_{n=p_{0}}^{p_{0}+M_{p}-1}T_{x,y}[n]$, and derive an upper bound on $\left|B_{x,y}\right|$.
\par As the p-trajectory is chosen about the $\left(N-1\right)^{\mathrm{th}}$ contour, the support set for its samples linearly increases from $p_0+1$ to $N$ coordinates, and then decreases to $p_0+1$ coordinates along the trajectory. The upper half of the p-trajectory contains the first $N-p_0$ coordinates sampled by the p-trajectory and the lower half contains the remaining $M_p-N+p_0$ of them.   
The partial sums in $B_{x,y}$ corresponding to the upper and lower components of the p-trajectory are defined as $B^{\mathrm{U}}_{x,y}=\sum_{n=p_{0}}^{N-1}T_{x,y}[n]$ and $B^{\mathrm{L}}_{x,y}=B_{x,y}-B^{\mathrm{U}}_{x,y}$.  
\setcounter{theorem}{0}
\begin{lemma}
For $p_0 \leq n  <N-1$, the function $T_{x,y}[n]$ satisfies 
\begin{align}
\label{boundTx}
\left|T_{x,y}[n+1]-e^{\mathrm{j} x}T_{x,y}[n]\right|&<\frac{2}{n+2} \,\, \mathrm{and}\\
\label{boundTy}
\left|T_{x,y}[n+1]-e^{\mathrm{j} y}T_{x,y}[n]\right|&<\frac{2}{n+2}.
\end{align}
\label{Lemma_T}
\end{lemma}
\textit{Proof.} In the upper-half of the p-trajectory, Swift-Link samples one of the $n+1$ coordinates on the $n^{\mathrm{th}}$  contour at random. For a uniform sampling distribution over the $n+1$ points, $T_{x,y}[n]$ is given by
\begin{equation}
T_{x,y}[n]=\frac{1}{n+1}\sum_{k=0}^{n}e^{\mathrm{j} (n-k)x+ky},
\label{Texplicit}
\end{equation}
for $p_0 \leq n  <N-1$. By triangle inequality, we have $\left|T_{x,y}[n]\right| \leq 1$ from \eqref{Texplicit}. The recursive equations for $T_{x,y}[n]$ can be derived from \eqref{Texplicit} as
\begin{align}
\label{Trecursx}
T_{x,y}[n+1]&=\left(\frac{n+1}{n+2}\right)e^{\mathrm{j} x}T_{x,y}[n]+\frac{1}{n+2}e^{\mathrm{j}(n+1)y}\,\, \textrm{and}\\
\label{Trecursy}
T_{x,y}[n+1]&=\left(\frac{n+1}{n+2}\right)e^{\mathrm{j} y}T_{x,y}[n]+\frac{1}{n+2}e^{\mathrm{j} (n+1)x}.
\end{align}
For $p_0 \leq n < N-1$, the absolute difference $\left|T_{x,y}[n+1]-e^{\mathrm{j} x}T_{x,y}[n]\right|$ can be computed from \eqref{Trecursx} as 
\begin{align}
\label{Trecx1}
\left|T_{x,y}[n+1]-e^{\mathrm{j} x}T_{x,y}[n]\right|&=\frac{1}{n+2}\left|e^{\mathrm{j} ny}-T_{x,y}[n]\right|.
\end{align}
The bound in \eqref{boundTx} follows from \eqref{Trecx1} using the triangle inequality and the fact that $\left|T_{x,y}[n]\right|\leq 1$. A similar argument can be made for \eqref{Trecursy} to obtain \eqref{boundTy}. 
\begin{lemma}
The absolute sum of the sequence $\{ T_{x,y}[n]\}_{n=p_0}^{p_0+M_p-1}$ corresponding to the upper-half and the total component of the p-trajectory are bounded as
\begin{align}
\label{Bup}
\left|B_{x,y}^{\mathrm{U}}\right|&\leq\left(2+\mathrm{log}\frac{N}{N+1-\left(M_{p}+1\right)/2}\right) \eta_{x,y} \,\,\, \mathrm{and}\\
\label{Ball}
\left|B_{x,y}\right|&\leq\left(4+2\mathrm{log}\frac{N}{N+1-\left(M_{p}+1\right)/2}\right) \eta_{x,y} \, , 
\end{align}
where $\eta_{x,y}=1/\mathrm{max}\{\left|\mathrm{sin}\left(x/2\right)\right|,\left|\mathrm{sin}\left(y/2\right)\right|\}$. 
\label{Lemma_B}
\end{lemma}
\textit{Proof.} The partial sum corresponding to the upper half of the p-trajectory can be written as 
\begin{align}
\label{Diff1_BU}
B_{x,y}^{\mathrm{U}}&=T_{x,y}\left[p_{0}\right]+\sum_{n=p_{0}}^{N-2}T_{x,y}[n+1], \mathrm{and}\\
\label{Diff2_BU}
e^{\mathrm{j}x}B_{x,y}^{\mathrm{U}}&=e^{\mathrm{j}x}T_{x,y}\left[N-1\right]+e^{\mathrm{j}x}\sum_{n=p_{0}}^{N-2}T_{x,y}[n].
\end{align} 
We use $H_I$ to denote the harmonic sum of the first $I$ natural numbers, i.e., $H_I=\sum^{I}_{i=1} 1/i$.
From \eqref{Diff1_BU} and \eqref{Diff2_BU}, we have 
\begin{align}
\nonumber
|(1-e^{\mathrm{j} x})B_{x,y}^{\mathrm{U}}|&= \Bigl| T_{x,y}[p_{0}]-e^{\mathrm{j} x}T_{x,y}[N-1]+ \\
\label{subtract1}
& \sum_{n=p_{0}}^{N-2}(T_{x,y}[n+1]-e^{ \mathrm{j} x}T_{x,y}[n])\Bigr|.
\end{align}
Now, we apply the triangle inequality to \eqref{subtract1} and upper bound $|(1-e^{\mathrm{j} x})B_{x,y}^{\mathrm{U}}|$ using $|T_{x,y}\left[ n \right] | \leq 1$ and  \eqref{boundTx}. Therefore, 
\begin{align}
\label{subtract2}
|(1-e^{\mathrm{j} x})B_{x,y}^{\mathrm{U}}| & \leq 2+2\sum_{n=p_{0}}^{N-2}\frac{1}{n+2} \,= \,2+ 2(H_N-H_{p_0+1}).
\end{align}
Using $\left|1-e^{\mathrm{j} x}\right|=2\left|\mathrm{sin}\left(x/2\right)\right|$ and  $\mathrm{log}\left(I+1\right)< H_I \leq 1+ \mathrm{log}I$, $\left|B_{x,y}^{\mathrm{U}}\right|$ can be bounded as 
\begin{equation}
\left|B_{x,y}^{\mathrm{U}}\right|\leq\left(2+\mathrm{log}\frac{N}{p_{0}+2}\right)\frac{1}{\left|\mathrm{sin}\left(x/2\right)\right|}.
\label{Bup_x}
\end{equation}
Similarly, the partial sum corresponding to the lower half of the p-trajectory can be bounded as 
\begin{equation}
\left|B_{x,y}^{\mathrm{L}}\right|\leq\left(2+\mathrm{log}\frac{N-1}{p_{0}+2}\right)\frac{1}{\left|\mathrm{sin}\left(x/2\right)\right|}.
\label{Blow_x}
\end{equation}
Applying triangle inequality to $B_{x,y}=B_{x,y}^{\mathrm{U}}+B_{x,y}^{\mathrm{L}}$, we have
\begin{align}
\left|B_{x,y}\right|&\leq\left(4+\mathrm{log}\frac{N\left(N-1\right)}{\left(p_{0}+2\right)^{2}}\right)\frac{1}{\left|\mathrm{sin}\left(x/2\right)\right|}\\
\label{Bineqx}
&\leq\left(4+2\mathrm{log}\frac{N}{p_{0}+2}\right)\frac{1}{\left|\mathrm{sin}\left(x/2\right)\right|}.
\end{align}
By the same structure of arguments along the $y$ dimension, we use \eqref{boundTy} to obtain
\begin{equation}
\left|B_{x,y}\right|\leq\left(4+2\mathrm{log}\frac{N}{p_{0}+2}\right)\frac{1}{\left|\mathrm{sin}\left(y/2\right)\right|}.
\label{Bineqy}
\end{equation}
Note that $p_0=N-(M_p+1)/2$. The inequalities in \eqref{Bineqx} and \eqref{Bineqy} can be combined to prove Lemma \ref{Lemma_B}.  
\par Now, we derive a sufficient condition for Swift-Link's p-component to satisfy the RIP with constant $\delta_K$ on an average. From \eqref{eq:expected_energy}, the deviation in the mean energy acquired by the p-trajectory is $|\mathbb{E}_{\Omega_{p}}[E]-\Vert\mathbf{S}\Vert _{F}^{2}|=\bigl|\mathrm{Re}\bigl\{\underset{k\neq\ell}{\sum}s_{k}s_{\ell}^{\ast}B_{2\pi\Delta x_{k,\ell}/N,2\pi\Delta y_{k,\ell}/N}\bigr\}\bigr|/ M_p$. Using $\left|\mathrm{Re}\{x\}\right|\leq\left|x\right|$ and the triangle inequality, this deviation can be upper bounded as
\begin{align}
\label{eq:conc_re_abs}
|\mathbb{E}_{\Omega_{p}}[E]-\Vert\mathbf{S}\Vert _{F}^{2}|& \leq\frac{1}{M_{p}}\underset{k\neq\ell}{\sum}\left|s_{k}\right|\left|s_{\ell}\right|\left|B_{2 \pi \Delta x_{k,\ell}/N,2 \pi \Delta y_{k,\ell}/N}\right|.
\end{align}
From \eqref{Ball}, it can be observed that the upper bound on $\left|B_{2 \pi \Delta x_{k,\ell}/N,2 \pi \Delta y_{k,\ell}/N}\right|$ is large when the spatial frequencies $\left(x_k,y_k\right)$ and $\left(x_{\ell},y_{\ell}\right)$ are close to each other. 
\par With $d_{\mathrm{min}}$ defined in Section \ref{sec:swlink_ana}, for any $k \neq \ell$ we have 
\begin{align}
\nonumber
\left|B_{2 \pi \Delta x_{k,\ell}/N,2 \pi \Delta y_{k,\ell}/N}\right| &\leq \frac{1}{\mathrm{sin}\left(\pi d_{\mathrm{min}}/N\right)}\times \\
&\left(4+2\mathrm{log}\frac{N}{N+1-(M_p+1)/2}\right).
\label{B_dmin}
\end{align}
The sum $\underset{k\neq\ell}{\sum}\left|s_{k}\right|\left|s_{\ell}\right|$ can be re-written as 
$\underset{k\neq\ell}{\sum}\left|s_{k}\right|\left|s_{\ell}\right|=\left\Vert \mathbf{S}\right\Vert _{\ell_{1}}^{2}-\left\Vert \mathbf{S}\right\Vert _{F}^{2}$.
Using the Cauchy-Schwartz inequality over the $K$ sparse matrix $\mathbf{S}$, i.e., $\left\Vert \mathbf{S}\right\Vert _{\ell_{1}}\leq\sqrt{K}\left\Vert \mathbf{S}\right\Vert _{F}$, we can write 
\begin{equation}
\underset{k\neq\ell}{\sum}\left|s_{k}\right|\left|s_{\ell}\right| \leq \left(K-1\right)\left\Vert \mathbf{S}\right\Vert^2 _{F}.
\label{eq:cauchy_sk_sl}
\end{equation} 
The inequalities in \eqref{eq:conc_re_abs}, \eqref{B_dmin} and \eqref{eq:cauchy_sk_sl} can be combined as 
\begin{align}
\nonumber
\left|\mathbb{E}_{\Omega_{p}}\left[E\right]-\left\Vert \mathbf{S}\right\Vert _{F}^{2}\right|&\leq\frac{\left(K-1\right)\left\Vert \mathbf{S}\right\Vert _{F}^{2}}{M_{p}\mathrm{sin}\left(\pi d_{\mathrm{min}}/N\right)} \times \\
&\left(4+2\mathrm{log}\frac{N}{N+1-\left(M_{p}+1\right)/2}\right).
\label{eq:expect_bound_fin}
\end{align}
By the definition of RIP in \eqref{eq:RIPdefn}, it can be noticed from \eqref{eq:expect_bound_fin} that the CS matrix corresponding to the p-trajectory satisfies RIP with constant $\delta_K$, on an average, when 
\begin{equation}
\frac{\left(K-1\right)}{M_{p}\mathrm{sin}\left(\pi d_{\mathrm{min}}/N\right)}\left(4+2\mathrm{log}\frac{N}{N+1-\left(M_{p}+1\right)/2}\right) \leq \delta_K,
\end{equation}
which proves Theorem \ref{theorem1}. 
\bibliographystyle{IEEEtran}
\bibliography{refs}
\vspace{10mm}
\begin{IEEEbiography}[{\includegraphics[width=1in,height=1.25in,clip,keepaspectratio]{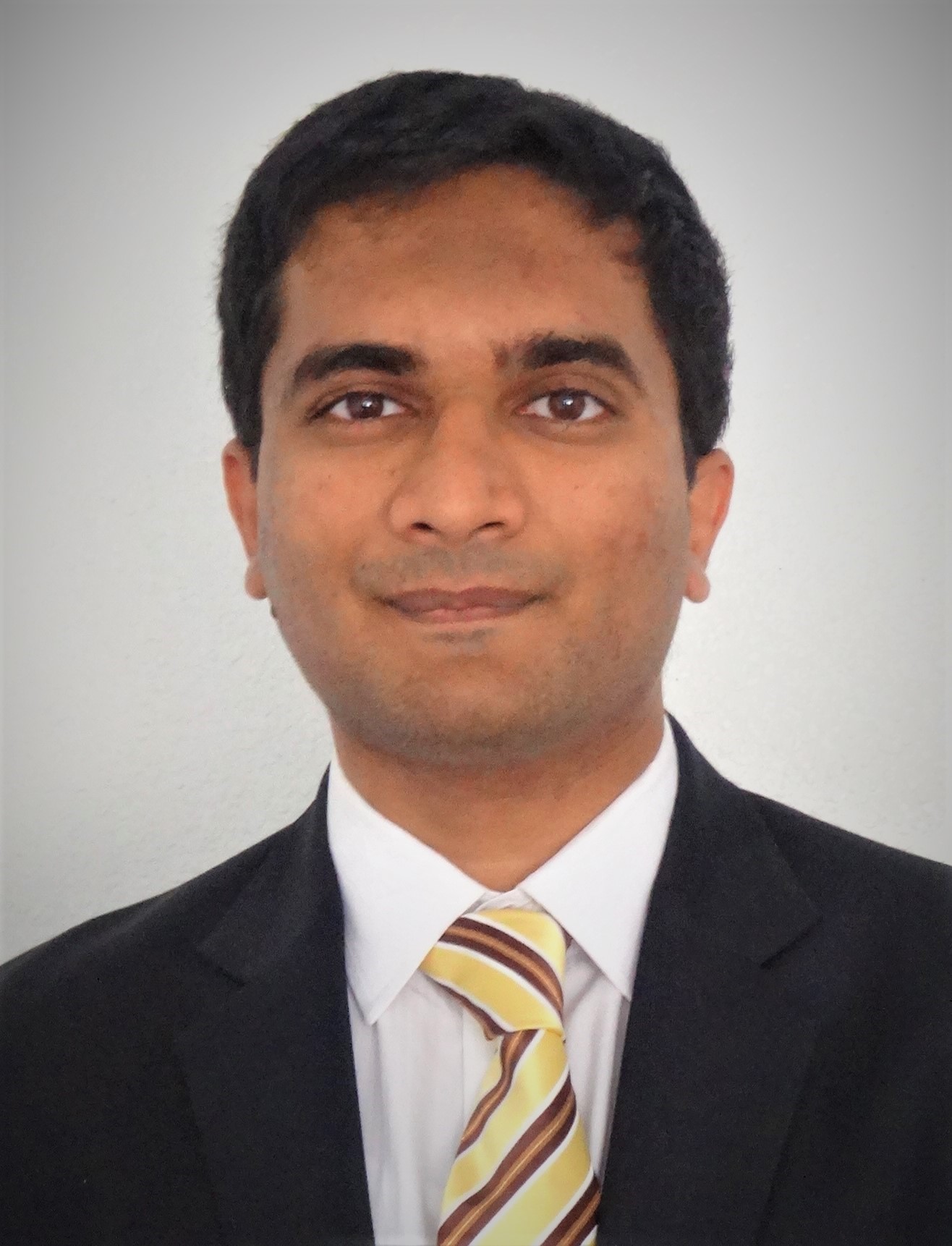}}]%
{Nitin Jonathan Myers}
received the B.Tech and M.Tech degrees in Electrical Engineering from the Indian Institute of Technology (IIT) Madras in 2016. He is currently pursuing the Ph.D. degree at the University of Texas at Austin. His research interests lie in the areas of wireless communications and signal processing. Mr. Myers received the 2018 Electrical and Computer Engineering research award from the Cockrell School of Engineering at The University of Texas at Austin. During his undergraduate days at IIT Madras, he received the DAAD WISE scholarship in 2014, and the Institute Silver Medal in 2016. 
\end{IEEEbiography}
\begin{IEEEbiography}[{\includegraphics[width=1in,height=1.25in,clip,keepaspectratio]{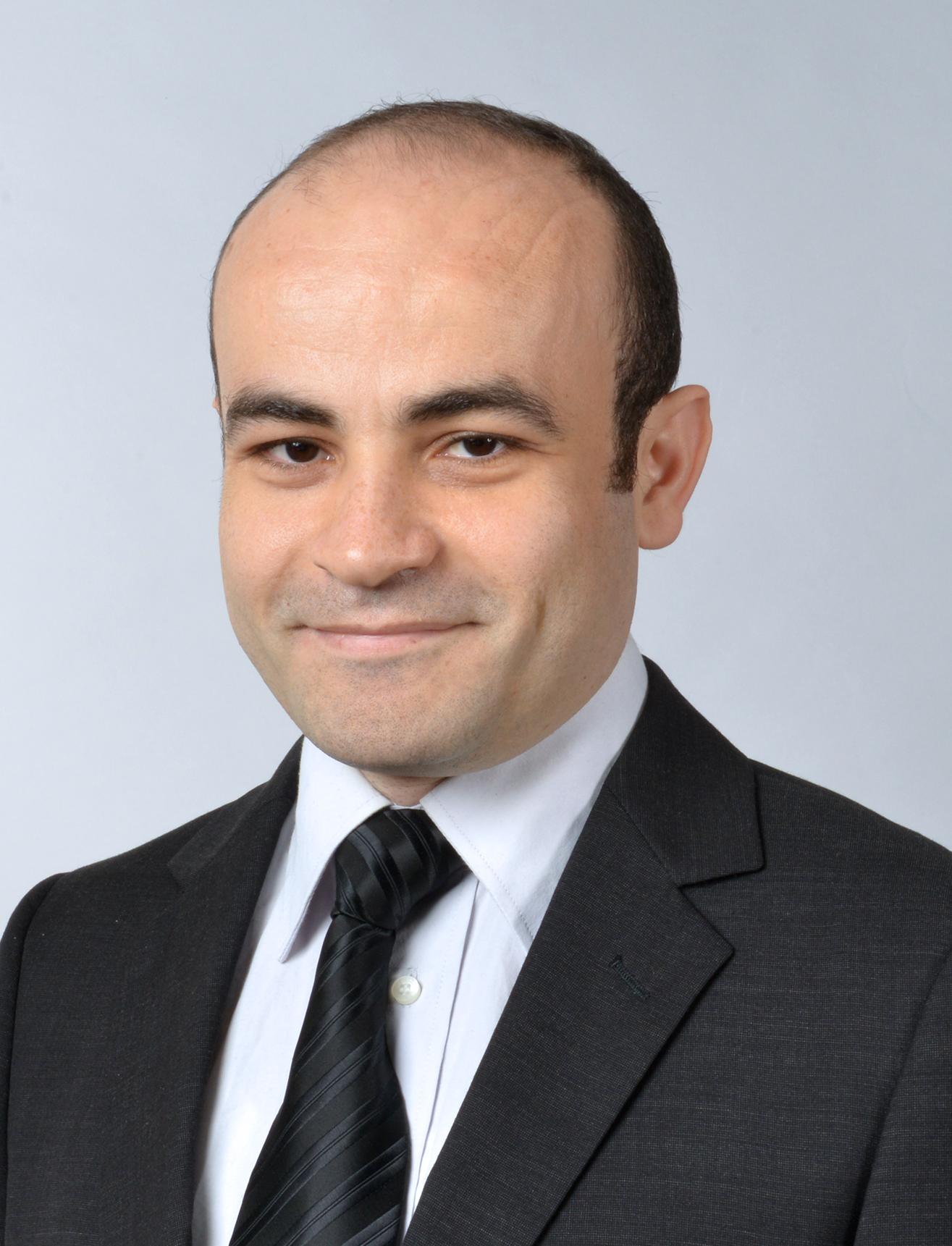}}]%
{Amine Mezghani}
(S'08, M'16) received the ``Dipl.-Ing." in Electrical 
Engineering from the Technische Universit\"at M\"unchen, Germany, and the
``Diplome d'Ing\'enieur" degree from the \'Ecole Centrale Paris, France,
both in 2006. He received the Ph.D. degree in Electrical Engineering
from the Technische Universit\"at M\"unchen in 2015. He was the recipient
of the Rohde \& Schwarz Outstanding Dissertation Award in 2016. In
Summer 2017, he joined the University of Texas at Austin as a
postdoctoral fellow. Prior to this, he has been working as
postdoctoral scholar in the Department of Electrical Engineering and
Computer Science of the University of California, Irvine, USA. His
current research interests are on millimeter-wave massive MIMO,
hardware constrained communication theory, and signal processing under
low-resolution analog-to-digital and digital-to-analog converters.
\end{IEEEbiography}
\begin{IEEEbiography}[{\includegraphics[width=1in,height=1.25in,clip]{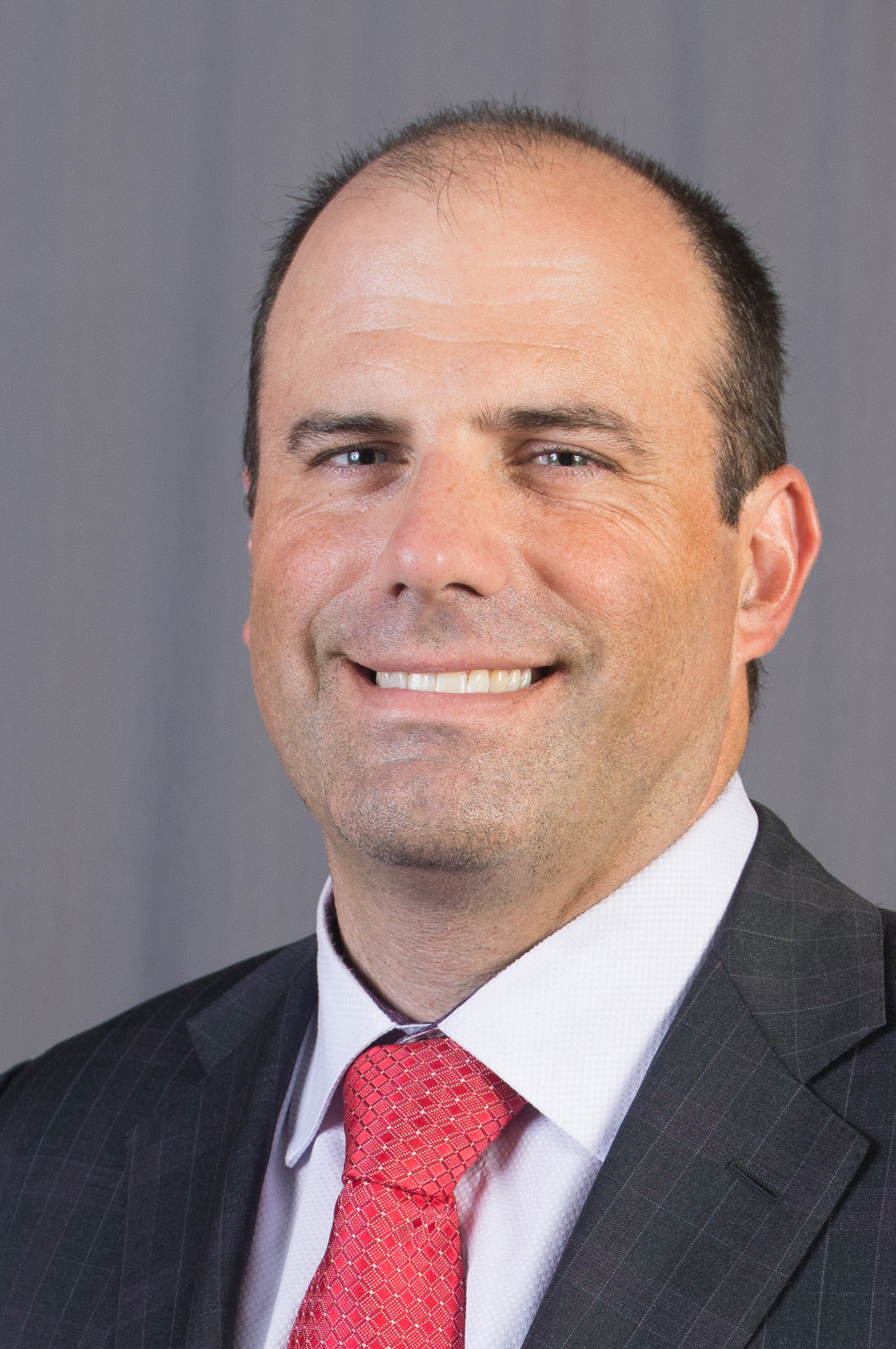}}]%
{Robert W. Heath Jr. } 
(S'96 - M'01 - SM'06 - F'11)  received the B.S. and M.S. degrees from the University of Virginia, Charlottesville, VA, in 1996 and 1997 respectively, and the Ph.D. from Stanford University, Stanford, CA, in 2002, all in electrical engineering. From 1998 to 2001, he was a Senior Member of the Technical Staff then a Senior Consultant at Iospan Wireless Inc, San Jose, CA where he worked on the design and implementation of the physical and link layers of the first commercial MIMO-OFDM communication system. Since January 2002, he has been with the Department of Electrical and Computer Engineering at The University of Texas at Austin where he is a Cullen Trust for Higher Education Endowed Professor, and is a Member of the Wireless Networking and Communications Group. He is also President and CEO of MIMO Wireless Inc. He authored “Introduction to Wireless Digital Communication” (Prentice Hall, 2017) and  “Digital Wireless Communication: Physical Layer Exploration Lab Using the NI USRP’' (National Technology and Science Press, 2012), and co-authored “Millimeter Wave Wireless Communications’' (Prentice Hall, 2014).
\par Dr. Heath has been a co-author of sixteen award winning conference and journal papers including  the 2010 and 2013 EURASIP Journal on Wireless Communications and Networking best paper awards, the 2012 Signal Processing Magazine best paper award, a 2013 Signal Processing Society best paper award, 2014 EURASIP Journal on Advances in Signal Processing best paper award, the 2014 and 2017 Journal of Communications and Networks best paper awards, the 2016 IEEE Communications Society Fred W. Ellersick Prize, the 2016 IEEE Communications and  Information Theory Societies Joint Paper Award, and the 2017 Marconi Prize Paper Award. He received the 2017 EURASIP Technical Achievement award. He was a distinguished lecturer in the IEEE Signal Processing Society and is an ISI Highly Cited Researcher. In 2017, he was selected as a Fellow of the National Academy of Inventors. He is also an elected member of the Board of Governors for the IEEE Signal Processing Society, a licensed Amateur Radio Operator, a Private Pilot, a registered Professional Engineer in Texas. He is currently Editor-in-Chief of IEEE Signal Processing Magazine.
\end{IEEEbiography}

\end{document}